\documentclass[11pt,usletter]{article}
\usepackage{color}
\usepackage{setspace}
\usepackage{times}
\usepackage{amsfonts}
\usepackage{amsmath}
\usepackage{amssymb}
\usepackage{amsthm}
\usepackage{marvosym}
\usepackage{graphicx}
\usepackage{small-caption}
\usepackage{sherstov}
   \usepackage{mathtime}
   \newcommand{\ccases}[1]{\begin{cases}#1\end{cases}}
   \newcommand{\PARENS}[1]{\left(#1\right)}
   \renewcommand{\BRACES}[1]{\left\{#1\right\}}
   \newcommand{\undercbrace}[1]{\underbrace{#1}}
   \newcommand{\LEFTRIGHT}[3]{\left#1{#3}\right#2}
   \newcommand{\rcbrace}{\}}
   \newcommand{\SQRT}[1]{\sqrt{#1}}
   \newcommand{\g}{g}

\renewcommand{\mathcal}[1]{\mathscr{#1}}

\input cyracc.def

\DeclareMathOperator{\Sign}{\widetilde{sgn}}
\renewcommand{\g}{g}
\newcommand{\op}{\text{\rm KP}}

\begin{document}
\sloppy

\title{The Intersection of Two Halfspaces Has \\
High Threshold Degree}

\date{}

\author{
{\sc Alexander~A.~Sherstov}\thanks{Department of Computer Sciences, 
University of Texas at Austin, TX 78757 USA.\newline
{\Large \Letter}~$\mathtt{sherstov@cs.utexas.edu}$}
}
\setcounter{page}{-1}
\maketitle
\thispagestyle{empty}
\pagestyle{plain}

\hyphenpenalty=200
\clubpenalty=100000
\widowpenalty=100000

\begin{abstract}
The \emph{threshold degree} of a Boolean function $f\colon\zoon\to\moo$
is the least degree of a real polynomial $p$ such that $f(x)\equiv\sign
p(x).$  We construct two halfspaces on $\zoon$ whose intersection
has threshold degree $\Theta(\sqrt n),$ an exponential improvement
on previous lower bounds. This solves an open problem due to
Klivans (2002) and rules out the use of perceptron-based
techniques for PAC learning the intersection of two halfspaces, a
central unresolved challenge in computational learning.  We also
prove that the intersection of two majority functions has
threshold degree $\Omega(\log n),$ which is tight and settles a
conjecture of O'Donnell and Servedio (2003).

Our proof consists of two parts. First, we show that for any
nonconstant Boolean functions $f$ and $g,$ the intersection $f(x)\wedge g(y)$
has threshold degree $O(d)$ if and only if $\|f-F\|_\infty +
\|g-G\|_\infty < 1$ for some rational functions $F,$ $G$ of
degree $O(d).$ Second, we settle the least degree
required for approximating a halfspace and a majority function to
any given accuracy by rational functions.

Our technique further allows us to make progress on
Aaronson's challenge (2008) and contribute strong direct product
theorems for polynomial representations of composed Boolean functions
of the form $F(f_1,...,f_n).$ In particular, we give an improved
lower bound on the approximate degree of the AND-OR tree.
\end{abstract}


\newpage
\thispagestyle{empty}
\tableofcontents
\newpage


\section{Introduction}
\label{sec:intro}

Representations of Boolean functions by real polynomials play
an important role in theoretical computer science, with applications
ranging from complexity theory to quantum computing and 
learning theory. The surveys
in~\mbox{\cite{beigel93polynomial-method,saks93slicing,buhrman-dewolf02DT-survey,dual-survey}} offer
a glimpse into the diversity of these results and techniques.
We study one such representation scheme known as
\emph{sign-representation}.  Specifically, fix a Boolean
function $f\colon X\to\moo$  for some finite set $X\subset \Re^n,$ such as
the hypercube $X=\moon.$
The \emph{threshold degree} of $f,$ denoted $\degthr(f),$ is the
least degree of a polynomial $p(x_1,\dots,x_n)$ such that
\[ f(x) = \sign p(x) \]
for each $x\in X.$ In other words, the threshold degree of $f$ is
the least degree of a real polynomial that represents $f$ in sign.

The formal study of this complexity measure and of sign-representations in
general began in 1969 with the seminal work of Minsky and
Papert~\cite{minsky88perceptrons}, who examined the threshold degree of several
common functions.  Since then, sign-representations have
found a variety of applications in theoretical computer science.
Paturi and Saks~\cite{paturi-saks94rational} and later Siu et
al.~\cite{siu-roy-kailath94rational} used Boolean functions with
high threshold degree to obtain size-depth trade-offs for threshold
circuits.  The well-known result, due to Beigel et
al.~\cite{beigel91rational}, that $\PP$ is closed under intersection
is also naturally interpreted in terms of threshold degree.  In
another development, Aspnes et al.~\cite{aspnes91voting} used the
notion of threshold degree and its relaxations to obtain oracle
separations for $\mathsf{PP}$ and to give an insightful new proof
of classical lower bounds for $\mathsf{AC}^0.$ Krause and
Pudl{\'a}k~\cite{krause94depth2mod,KP98threshold} used random
restrictions to show that the threshold degree gives lower
bounds on the \emph{weight} and \emph{density} of perceptrons and
their generalizations, which are well-studied computational models.

Learning theory is another area in which the threshold degree of
Boolean functions is of considerable interest.  Specifically,
functions with low threshold degree can be efficiently PAC
learned under arbitrary distributions via linear programming.
The current fastest algorithm for PAC learning polynomial-size
DNF formulas, due to Klivans and Servedio~\cite{KS01dnf}, is an
illustrative example: it is based precisely on an upper bound on
the threshold degree of this concept class.

The threshold degree has recently become a versatile
tool in communication complexity.  The starting point in this line
of work is the Degree/Discrepancy
Theorem~\cite{sherstov07ac-majmaj,sherstov07quantum}, which states
that any Boolean function with high threshold degree induces a
communication problem with low \emph{discrepancy} and thus high
communication complexity in almost all models.  This result was used 
in~\cite{sherstov07ac-majmaj} to show the optimality of Allender's
simulation of $\AC^0$ by majority
circuits~\cite{allender89ac0tc0}, thus solving an open
problem of Krause and Pudl{\'ak}~\cite{krause94depth2mod}.  
Known lower bounds on the threshold degree have 
played an important role in recent
progress~\cite{sherstov07symm-sign-rank,RS07dc-dnf} on
\emph{unbounded-error} communication complexity, which is
considerably more powerful than the models above.

%

In summary, the threshold degree has a variety of applications in
circuit complexity, learning theory, and communication
complexity. 
Nevertheless, analyzing the threshold degree has remained a difficult
task, and Minsky and Papert's \emph{symmetrization} technique
from 1969 has been essentially the only method available.
Unfortunately, symmetrization only applies to
symmetric Boolean functions and certain derivations thereof.  In
a recent tutorial presented at the FOCS'08
conference, Aaronson~\cite{aaronson08tutorial} re-posed the
challenge of developing new analytic techniques for
multivariate real polynomials that represent Boolean functions.
We make significant progress on this challenge in the context of
sign-representation, contributing a number of strong direct
product theorems for the threshold degree.  As an application, we
construct two halfspaces on $\zoon$ whose intersection has
threshold degree $\Omega(\sqrt n),$ which solves an open problem
due to Klivans~\cite{klivans-thesis} and rules out the use of
perceptron-based techniques for PAC learning the intersection of
even two halfspaces (a central unresolved challenge in
computational learning theory). We give a detailed description of
our results in
Sections~\ref{sec:general-compositions}--\ref{sec:results-hshs},
followed by a discussion of our techniques in
Section~\ref{sec:techniques}.

\subsection{Results for general compositions}
\label{sec:general-compositions}

Our first result is a general direct product theorem for the threshold
degree of composed functions.

\begin{theorem}[Threshold degree]
\label{thm:thrdeg-dp}
Consider functions
$f\colon X\to\moo$ and $F\colon \mook\to\moo,$ where
$X\subset\Re^n$ is a finite set. Then
\begin{align*}
\degthr(F(f,\dots,f)) \geq \degthr(F)\degthr(f).
\end{align*}
\end{theorem}

Theorem~\ref{thm:thrdeg-dp} gives the best possible
lower bound that depends on $\degthr(F)$ and $\degthr(f)$ alone. In
particular, the bound is tight whenever $F=\PARITY$ or
$f=\PARITY.$ To our knowledge, the only previous direct product
theorem of any kind for the threshold degree was the XOR lemma
in~\cite{odonnell03degree}, which states that the XOR of $k$
copies of a given function $f\colon X\to\moo$ has threshold
degree $k\degthr(f).$

We are able to generalize Theorem~\ref{thm:thrdeg-dp} to the notion
of \emph{$\epsilon$-approximate degree} $\degeps(F),$ which is the
least degree of a real polynomial $p$ with $\|F-p\|_\infty\leq\epsilon.$
This notion plays a fundamental role in complexity theory, learning
theory, and quantum computing
and was also re-posed as an analytic challenge in Aaronson's
tutorial~\cite{aaronson08tutorial}.  We have:

\begin{theorem}[Approximate degree]
\label{thm:main-approx-dp}
Fix functions
$f\colon X\to\moo$ and $F\colon \mook\to\moo,$ where
$X\subset\Re^n$ is a finite set. Then for 
$0<\epsilon<1,$  
\begin{align*}
\deg_{\epsilon}(F(f,\dots,f)) \geq \deg_{\epsilon}(F)\degthr(f).
\end{align*}
\end{theorem}

Again, Theorem~\ref{thm:main-approx-dp} gives the best lower
bound that depends on $\deg_\eps(F)$ and $\degthr(f)$ alone. For
example, the stated bound is tight for any function $F$ when $f=\PARITY.$
In Section~\ref{sec:dp}, we prove various other results involving
bounded-error and small-bias approximation, as well as
compositions of the form $F(f_1,\dots,f_k)$ where $f_1,\dots,f_k$
may all be distinct. 

We use Theorem~\ref{thm:main-approx-dp} to
obtain an improved lower bound on the approximate degree
of the well-studied AND-OR tree, given by 
\begin{align}
f(x)=\bigvee_{i=1}^n
\bigwedge_{j=1}^n x_{ij}.
\label{eqn:and-or-def}
\end{align}
Prior to this work, the best lower bound was
$\Omega(n^{0.66\dots}),$ due to
Ambainis~\cite{ambainis05collision}.  Preceding it were lower
bounds of $\Omega(\sqrt n)$ due to Nisan and
Szegedy~\cite{nisan-szegedy94degree} and $\Omega(\sqrt{n\log n})$
due to Shi~\cite{shi-linear}. We improve the standing lower bound
from $\Omega(n^{0.66\dots})$ to $\Omega(n^{0.75}),$ the best
upper bound being $O(n)$
due to H{\o}yer et al.~\cite{hoyer-mosca-dewolf03and-or-tree}.

\begin{theorem}[AND-OR Tree]
\label{thm:main-and-or}
Define $f\colon \moo^{n^2}\to\moo$ by
\textup{(\ref{eqn:and-or-def})}.
Then 
\begin{align*}
\deg_{1/3}(f) = \Omega(n^{0.75}).
\end{align*}
\end{theorem}

\noindent
Furthermore, the proof of Theorem~\ref{thm:main-and-or} is simpler 
and more modular than the previous lower bound~\cite{ambainis05collision},
which was based on the collision and element distinctness problems.

\subsection{Results for specific compositions}
\label{sec:results-conj}

While Theorems~\ref{thm:thrdeg-dp}
and~\ref{thm:main-approx-dp} give the best lower bounds that depend on
$\degthr(F),$ $\degthr(f),$ and $\degeps(F)$ alone, much stronger
lower bounds can be derived in some cases by exploiting additional
structure of $F$ and $f.$
%
%
%
Consider the special but illustrative
case of the conjunction of two functions.
In other words, we are given functions $f\colon X\to\moo$ and
$\g\colon Y\to\moo$ for some finite sets $X,Y\subset\Re^n$ and would like
to determine the threshold degree of their conjunction, $(f\wedge
\g)(x,y) = f(x)\wedge \g(y).$ A simple and elegant method for
sign-representing $f\wedge \g,$ due to Beigel et
al.~\cite{beigel91rational}, is to use rational approximation.
Specifically, let $p_1(x)/q_1(x)$ and $p_2(y)/q_2(y)$ be rational
functions of degree $d$ that approximate $f$ and $\g,$ respectively,
in the following sense:
\begin{align}
\label{eqn:approx-f-g}
   \max_{x\in X} \left| f(x) - \frac{p_1(x)}{q_1(x)}\right| \,+ \,
   \max_{y\in Y} \left| \g(y) - \frac{p_2(y)}{q_2(y)}\right| \,<\, 1. 
\end{align}
Letting $-1$ and $+1$ correspond to ``true'' and ``false,'' respectively,
we obtain:
\begin{align}
\label{eqn:pq}
 f(x)\wedge \g(y) &\equiv \sign\{1 + f(x)+\g(y)\}
 \rule{0mm}{7mm}\equiv \sign\BRACES{ 1 + \frac{p_1(x)}{q_1(x)} +
\frac{p_2(y)}{q_2(y)} }.
\end{align}
Multiplying the last expression in braces by the positive quantity
$q_1(x)^2q_2(y)^2$ gives
\begin{multline*}
\quad f(x)\wedge \g(y) \equiv \sign\left\{ q_1(x)^2q_2(y)^2 \right.
\\\left.+p_1(x)q_1(x)q_2(y)^2 + p_2(y)q_1(x)^2q_2(y)\right\},\quad
\end{multline*}
whence $\degthr(f\wedge \g) \leq 4d.$ In summary, if $f$ and $\g$ can
be approximated as in (\ref{eqn:approx-f-g}) by rational functions of
degree at most $d,$ then the conjunction $f\wedge \g$ has
threshold degree at most $4d.$

It is natural to ask whether there exists a better 
construction.  After all, given a sign-representing polynomial
$p(x,y)$ for $f(x)\wedge \g(y),$ there is no reason to expect
that $p$ arises from the sum of two independent rational functions
as in~(\ref{eqn:pq}).  Indeed, $x$ and $y$ can be tightly coupled
inside $p(x,y)$ and can interact in complicated ways. Our next
result is that, surprisingly, no such interactions
can beat the simple construction above.  In other words, the 
sign-representation based on rational functions always achieves
the optimal degree, up to a small constant factor.

\begin{theorem}[Conjunctions of functions]
\label{thm:main-two}
Let $f\colon X\to\moo$ and $\g\colon Y\to\moo$ be given functions, where
$X,Y\subset\Re^n$ are arbitrary finite sets.
Assume that $f$ and $\g$ are not identically false. Let 
$d=\degthr(f\wedge \g).$ Then there exist degree-$4d$ rational
functions
\[ \frac{p_1(x)}{q_1(x)}, \quad 
   \frac{p_2(y)}{q_2(y)} \]
that satisfy \textup{(\ref{eqn:approx-f-g})}.
\end{theorem}

Via repeated applications of Theorem~\ref{thm:main-two}, we are
able to obtain analogous results for conjunctions $f_1\wedge f_2\wedge
\cdots \wedge f_k$ for any Boolean functions $f_1,f_2,\dots,f_k$
and any $k.$ Our results further extend to compositions $F(f_1,\dots,f_k)$
for various $F$ other than $F=\AND,$ such as 
halfspaces and read-once AND/OR/NOT formulas.
We defer a more detailed description of these extensions
to Section~\ref{sec:h}, limiting this overview to the following
representative special case.

\begin{theorem}[Extension to multiple functions]
\label{thm:main-h}
Let $f_1,f_2,\dots,f_k$ be nonconstant Boolean functions on finite 
sets $X_1,X_2,\dots,X_k\subset\Re^n,$ respectively.  Let
$F\colon \mook\to\moo$ be a halfspace or a read-once AND/OR/NOT 
formula.  Assume that $F$ depends on
all of its $k$ inputs and that the composition $F(f_1,f_2,\dots,f_k)$
has threshold degree $d.$ Then there is a degree-$D$ rational
function $p_i/q_i$ on $X_i,$ $i=1,2,\dots,k,$ such that
\[ \sum_{i=1}^k \;
   \max_{x_i\in X_i} \left| f_i(x_i) -
   \frac{p_i(x_i)}{q_i(x_i)}\right|<1,\]
where $D=8d\log2k.$
\end{theorem}

\noindent
Theorem~\ref{thm:main-h} is close to
optimal. For example, when $F=\AND,$ the upper bound on
$D$ is tight up to a factor of $\Theta(k\log k)$; for all $F$
in the statement of the theorem, it is tight up to a polynomial
in $k.$ See Remark~\ref{rem:tightness} for details.

Theorems~\ref{thm:main-two} and~\ref{thm:main-h} contribute a
strong technique for proving lower bounds on the threshold degree,
via rational approximation.  Prior to this paper, it was a substantial
challenge to analyze the threshold degree even for
compositions of the form $f\wedge \g.$ Indeed, we are only aware of
the work in~\cite{minsky88perceptrons,odonnell03degree}, where
the threshold degree of $f\wedge \g$ was studied for the special
case $f=\g=\text{\sc majority}.$ The main difficulty in those previous
works was analyzing the unintuitive interactions between $f$ and
$\g.$ Our results remove this difficulty, even in the general
setting of compositions $F(f_1,f_2,\dots,f_k)$ for arbitrary
$f_1,f_2,\dots,f_k$ and various combining functions $F.$ Specifically,
Theorems~\ref{thm:main-two} and~\ref{thm:main-h} make it possible
to study the base functions $f_1,f_2,\dots,f_k$ individually,
in isolation. Once their rational approximability is understood,
one immediately obtains lower bounds on the threshold degree of
$F(f_1,f_2,\dots,f_k).$

\subsection{Results for intersections of two halfspaces}
\label{sec:results-hshs}

As an application of our direct product theorems in
Section~\ref{sec:results-conj}, we obtain the first strong lower bounds on 
the threshold degree of
intersections of halfspaces, i.e., intersections of functions of
the form 
$f(x)=\sign(\sum \alpha_i x_i-\theta)$ for some reals
$\alpha_1,\dots,\alpha_n,\theta.$ 
In light of Theorem~\ref{thm:main-two}, this task amounts to
proving that rational functions of low degree cannot approximate
a given halfspace. We accomplish this in the following theorem,
where the notation $\rdeg_\eps(f)$ stands for the least degree of a
rational function $A$ with $\|f-A\|_\infty\leq\epsilon.$

\begin{theorem}[Approximation of a halfspace]
\label{thm:main-approx-hs}
Let $f\colon\moo^{n^2}\to\moo$ be given by
\begin{align}
\label{eqn:halfspace-defined}
f(x) = \sign\PARENS{1 + \sum_{i=1}^{\phantom{A}n\phantom{A}}
   \sum_{j=1}^n 2^i x_{ij}}.
\end{align}
Then for $1/3<\epsilon<1,$ 
\begin{align*}
\rdeg_\epsilon(f) = \Theta\PARENS{1 + \frac
{\phantom{A}n\phantom{A}}{\log\{1/(1-\epsilon)\}}}.
\end{align*}
Furthermore, for all $\epsilon>0,$
\begin{align*}
\rdeg_\epsilon(f) \leq 64 n\lceil \log_2 n \rceil +1.
\end{align*}
\end{theorem}

The function (\ref{eqn:halfspace-defined}) is known 
as the \emph{canonical halfspace}.
Thus, Theorem~\ref{thm:main-approx-hs} shows that a rational
function of degree $\Theta(n)$ is necessary and sufficient for
approximating the canonical halfspace within $1/3.$ The upper bound
in this theorem follows readily from classical work by
Newman~\cite{newman64rational}, and it is 
the lower bound that has required of us technical
novelty and effort.  The best previous degree lower bound for
constant-error approximation for any halfspace was $\Omega(\log
n/\log\log n),$ obtained implicitly in~\cite{odonnell03degree}.
We complement Theorem~\ref{thm:main-approx-hs} with a full solution for 
another common halfspace, the majority function.

\begin{theorem}[Approximation of majority]
\label{thm:main-approx-maj}
Let $\MAJ_n\colon \moon\to\moo$ denote the majority function. Then
\begin{align*}
\rdeg_\epsilon(\MAJ_n)= 
\ccases{ 
\displaystyle 
\Theta\PARENS{
\log \BRACES{\frac{2n}{\log(1/\epsilon)}}
\cdot 
\log \frac1\epsilon
},
	&\qquad 2^{-n}<\epsilon<1/3,\\
\rule{0mm}{10mm}
\displaystyle 
\Theta\PARENS{1 + \frac{\log n}{\log\{1/(1-\epsilon)\}}},
	&\qquad 1/3\leq \epsilon<1.
}
\end{align*}
\end{theorem}

\noindent
Again, the upper bound in Theorem~\ref{thm:main-approx-maj} is
relatively straightforward.  Indeed, an upper bound
of $O(\log\{1/\epsilon\}\log n)$ for $0<\epsilon<1/3$ was known and
used in the complexity literature long before our 
work~\cite{paturi-saks94rational,
siu-roy-kailath94rational, beigel91rational, KOS:02}, and we only
somewhat tighten that upper bound and extend it to all $\epsilon.$
Our primary contribution in Theorem~\ref{thm:main-approx-maj},
then, is a matching \emph{lower} bound on the degree, which requires
considerable effort. The closest previous line of research concerns
\emph{continuous} approximation of the sign function on
$[-1,-\epsilon]\cup[\epsilon,1],$ which unfortunately gives no
insight into the discrete case.  For example, the lower bound derived
by Newman~\cite{newman64rational} in the continuous setting is based
on the integration of relevant rational functions with respect to
a suitable weight function, which has no meaningful discrete analogue.  We
discuss our solution in greater detail at the end of the introduction.

Our first application of these lower bounds for rational
approximation is to construct an intersection of two halfspaces
with high threshold degree.  In what follows, the symbol $f\wedge f$
denotes the conjunction of two independent copies of a given
function $f.$

\begin{theorem}[Intersection of two halfspaces]
\label{thm:main-sign-hs}
Let $f\colon\moo^{n^2}\to\moo$ be given by
\textup{(\ref{eqn:halfspace-defined})}.
Then
\begin{align*}
\degthr(f\wedge f)  = \Omega(n).
\end{align*}
\end{theorem}

The lower bound in Theorem~\ref{thm:main-sign-hs} is tight and
matches the construction by Beigel~et
al.~\cite{beigel91rational}.  Prior to our work,
only an $\Omega(\log n/\log\log
n)$ lower bound was known on the threshold degree of the
intersection of two halfspaces, due to O'Donnell and
Servedio~\cite{odonnell03degree}, preceded in turn by an
$\omega(1)$ lower bound of Minsky and
Papert~\cite{minsky88perceptrons}.  Note that
Theorem~\ref{thm:main-sign-hs} requires the difficult part of
Theorem~\ref{thm:main-approx-hs}, namely, the lower bound for the
rational approximation of a halfspace.

Theorem~\ref{thm:main-sign-hs} solves an open problem in
computational learning theory, due to
Klivans~\cite{klivans-thesis}.  In more detail, recall that
Boolean functions with low threshold degree can be efficiently
PAC learned under arbitrary distributions, by expressing an
unknown function as a perceptron with unknown weights and solving
the associated linear program~\cite{KS01dnf,KOS:02}.  Now, a
central challenge in the area is PAC learning the intersection of
two halfspaces under arbitrary distributions, which remains
unresolved despite much effort and solutions to some restrictions
of the problem, e.g.,~\cite{KwekPitt:98, Vempala:97, KOS:02,
KlivansServedio:04coltmargin}.  Prior to this paper, it was 
unknown whether intersections of two halfspaces on $\zoon$ are amenable to
learning via perceptron-based techniques.  Specifically,
Klivans~\cite[\S7]{klivans-thesis} asked for a lower bound of
$\Omega(\log n)$ or better on the threshold degree of the
intersection of two halfspaces.  We solve this problem with a
lower bound of $\Omega(\sqrt n),$ thereby ruling out the use of
perceptron-based techniques for learning the intersection of two
halfspaces in subexponential time.  To our knowledge,
Theorem~\ref{thm:main-sign-hs} is the first unconditional,
structural lower bound for PAC learning the intersection of two
halfspaces; all previous hardness results for the problem were
based on complexity-theoretic assumptions~\cite{blum92trainingNN,
ABFKP:04, focs06hardness, khot-saket08hs-and-hs}.  We 
complement Theorem~\ref{thm:main-sign-hs} as follows.

\begin{theorem}[Mixed intersection]\quad 
\label{thm:main-sign-mixed}
Let $f\colon\moo^{n^2}\to\moo$ be given by
\textup{(\ref{eqn:halfspace-defined})}.  Let $\g\colon\moo^{\lceil\sqrt
n\rceil}\to\moo$ be the majority function.  Then
\begin{align*}
\degthr(f\wedge\g) = \Theta(\sqrt n).
\end{align*}
\end{theorem}

\noindent
In words, even if one of the halfspaces in Theorem~\ref{thm:main-sign-hs}
is replaced by a majority function, the threshold degree will remain
high, resulting in a challenging learning problem. Finally, we have:

\begin{theorem}[Intersection of two majorities]
\label{thm:main-sign-maj}
Consider the majority function 
$\MAJ_n\colon\moon\to\moo.$
Then
\begin{align*}
\degthr(\MAJ_n \wedge \MAJ_n)  = \Omega(\log n).   
\end{align*}
\end{theorem}

\noindent
Theorem~\ref{thm:main-sign-maj} is tight, matching the
construction of Beigel et al.~\cite{beigel91rational}.  It
settles a conjecture of O'Donnell and
Servedio~\cite{odonnell03degree}, who gave a lower bound of
$\Omega(\log n/\log \log n)$ with completely different techniques
and conjectured that the true answer was $\Omega(\log n).$
Theorems~\ref{thm:main-sign-hs}--\ref{thm:main-sign-maj} are of
course also valid for disjunctions rather than
conjunctions.  Furthermore, Theorems~\ref{thm:main-sign-hs}
and~\ref{thm:main-sign-maj} remain tight with respect to
conjunctions of any constant number of functions.

Finally, we believe that the lower bounds for rational approximation in
Theorems~\ref{thm:main-approx-hs} and~\ref{thm:main-approx-maj}
are of independent interest.  Rational functions are classical
objects with various applications in
theoretical computer science~\cite{beigel91rational,
paturi-saks94rational,
siu-roy-kailath94rational,
KOS:02,
aaronson05postselection},
and yet our ability to prove strong lower bounds for
the rational approximation of Boolean functions has seen little
progress since the seminal work in 1964 by
Newman~\cite{newman64rational}.  To illustrate some of the
counterintuitive phenomena involved in rational approximation,
consider the familiar function $\OR_n\colon\zoon\to\moo,$ given
by $\OR_n(x)=1 \Leftrightarrow x=0.$ A well-known result of Nisan
and Szegedy~\cite{nisan-szegedy94degree} states that
$\deg_{1/3}(f)=\Theta(\sqrt n),$ meaning that a polynomial of
degree $\Theta(\sqrt n)$ is required for approximation within
$1/3.$ At the same time, we claim that $\rdeg_\epsilon(f)=1$ for
all $0<\epsilon<1.$ Indeed, let
\begin{align*}
A_M(x) = \frac{1 - M\sum x_i}{1+M\sum x_i}.
\end{align*}
Then $\|f-A_M\|_\infty\to0$ as $M\to\infty.$
This example illustrates that proving lower bounds for rational
functions can be a difficult and unintuitive task. We hope that
Theorems~\ref{thm:main-approx-hs} and~\ref{thm:main-approx-maj}
in this paper will spur further progress on the rational
approximation of Boolean functions.

\subsection{Our techniques \label{sec:techniques}}

We use one set of techniques to obtain our direct product
theorems for the threshold degree
(Sections~\ref{sec:general-compositions}
and~\ref{sec:results-conj}) and another, unrelated set of
techniques to analyze the rational approximation of halfspaces
(Section~\ref{sec:results-hshs}). We will give a separate overview
of the technical development in each case.

\bigskip

\noindent
\emph{Direct product theorems.}
In symmetrization, one takes an assumed multivariate polynomial $p$ that
sign-represents a given symmetric function and converts $p$ into
a univariate polynomial, which is amenable to direct analysis. No
such approach works for the function compositions of this paper,
whose sign-representing polynomials can have complicated
structure and will not simplify in a meaningful way. This leads
us to pursue a completely different approach.

Specifically, our results are based on a thorough
study of the \emph{linear programming dual} of the sign-representation
problems at hand. The challenge in our work is to bring out, through
the dual representation, analytic properties that will obey a
direct product theorem. Depending on the context 
(Theorem~\ref{thm:thrdeg-dp},~\ref{thm:main-approx-dp},
or~\ref{thm:main-two}), the property in question can be nonnegativity,
correlation, orthogonality, certain quotient structure, or a
combination of several of these. A strength of this approach
is that it works with the sign-representation problem itself (over which
we have considerable control) rather than an assumed sign-representing
polynomial (whose structure we can no longer control in a meaningful
way). We are confident that this approach will find other
applications.

As a concrete illustration, we briefly describe the idea behind
Theorem~\ref{thm:main-two}. The dual object with which we work there
is a certain problem of finding, in the positive spans of two given
matrices, two vectors whose corresponding entries have comparable
magnitude.  By an analytic argument, we are able to prove that this
intermediate problem has the sought direct-product property, giving
the missing link between sign-representation and rational approximation.
Thus, by working with the dual, we \emph{implicitly} decompose any
sign-representation $p(x,y)$ of the function $f(x)\wedge \g(y)$
into individual rational approximants for $f$ and $\g,$ regardless
of how tightly the $x$ and $y$ parts are coupled inside $p.$

\bigskip

\noindent
\emph{Rational approximation.}
Our proof of Theorem~\ref{thm:main-approx-hs}
is built around two key ideas.
The first is a new technique for placing
lower bounds on the degree of a given polynomial
$p\in\Re[x_1,x_2,\dots,x_n]$ with prescribed approximate behavior,
whereby one constructs a degree-nonincreasing linear map
$M\colon\Re[x_1,x_2,\dots,x_n]\to\Re[x]$ and argues that $Mp$ has
high degree. This technique is crucial to proving
Theorem~\ref{thm:main-approx-hs}, which 
is not amenable to standard techniques such
as symmetrization. As applied in this work, the technique amounts
to constructing random variables $\xbold_1,\xbold_2,\dots,\xbold_n$ in
Euclidean space that, on the one hand, satisfy the linear dependence
$\sum 2^i\xbold_i\equiv \zbold$ for a suitably fixed vector $\zbold$ and,
on the other hand, in expectation look independent to any low-degree
polynomial $p\in\Re[x_1,x_2,\dots,x_n].$ We pass, then, from $p$ to a
univariate polynomial by observing that 
$\Exp[p(\xbold_1,\dots,\xbold_n)]=q(\zbold)$ for some univariate polynomial
$q$ of degree no greater than the degree of $p.$ This technique is
a substantial departure from previous methods and shows promise
on other problems involving approximation by polynomials or rational
functions.

Second, we are able to prove that the rational approximation of the
sign function has a self-reducibility property on the discrete
domain. More specifically, we are able to give an explicit solution
to the \emph{dual} of the rational approximation problem by
distributing the nodes as in known positive results. What makes
this program possible in the first place is our ability to zero out
the dual object on the complementary domain, which is where the
above map $M\colon\Re[x_1,x_2,\dots,x_n]\to\Re[x]$ plays a crucial
role.  This dual approach, too, departs entirely from previous
analyses. In particular, recall that Newman's lower-bound analysis
is specialized to the continuous domain and does not extend
to the setting of Theorem~\ref{thm:main-approx-maj}, let alone
Theorem~\ref{thm:main-approx-hs}.


%

\subsection*{Recent progress}
A recent follow-up paper~\cite{sherstov09opthshs} proves that
the intersection of two halfspaces on $\zoon$ has threshold
degree $\Theta(n),$ improving on the lower bound of $\Omega(\sqrt
n)$ in this work.
We have also learned that the inequality
$\deg_\eps(F(f,\dots,f))\geq\degeps(F)\degthr(f)$ was derived
independently by Lee~\cite{lee09formulas} in a recent work on
read-once Boolean formulas.

\section{Preliminaries}
Throughout this work, the symbol $t$ refers to a real variable, whereas
$u,$ $v,$ $w,$ $x,$ $y,$ $z$ refer to vectors in $\Re^n$ and in
particular in $\moon.$ We adopt the following standard definition
of the sign function:
\begin{align*}
  \sign t = 
  \ccases{
  -1, &t<0, \\
  0,  &t=0, \\
  1, &t>0.
  }
\end{align*}
We will also have occasion to use the following modified sign function:
\[ 
  \Sign t = 
  \ccases{
  -1, &t<0, \\
  1,  &t\geq0.
  }
\]
Equations and inequalities involving vectors in $\Re^n,$
such as $x<y$ or $x\geq0,$ are to be interpreted component-wise, as
usual.

Throughout this manuscript, we view Boolean functions as mappings
$f\colon X\to\moo$ for some finite set $X,$ where $-1$ and $+1$ correspond to
``true'' and ``false,'' respectively.  If $\mu_1,\dots,\mu_k$ are
probability distributions on finite sets $X_1,\dots, X_k,$ respectively,
then $\mu_1\times\cdots\times\mu_k$ stands for the probability
distribution on $X_1\times\cdots\times X_k$ given by
\[ (\mu_1\times\cdots\times\mu_k)(x_1,\dots,x_k) = \prod_{i=1}^k\mu_i(x_i).
\]
The majority function on $n$ bits, $\MAJ_n\colon\moon\to\moo,$ is given
by 
\[ \MAJ_n(x) = \ccases{
1, & \sum x_i > 0,\\
-1, &\text{otherwise.}
}\]%
The symbol $P_k$ stands for the family of all univariate real polynomials
of degree up to $k.$ The following combinatorial identity is well-known.

\begin{fact}
\label{fact:comb}
For every integer $n\geq1$ and every polynomial $p\in P_{n-1},$
\[ \sum_{i=0}^n {n\choose i} (-1)^i p(i) = 0. \]
\end{fact}

\noindent 
This fact can be verified by repeated differentiation
of the real function 
\[(t-1)^n = \sum_{i=0}^n {n\choose i} (-1)^{n-i} t^i\]
at $t=1,$ as explained in~\cite{odonnell03degree}.

For a real function $f$ on a finite set $X,$ we write
$\|f\|_\infty=\max_{x\in X} |f(x)|.$
For a subset $X\subseteq\Re^n,$ we adopt the notation $-X=\{-x : x\in
X\}.$ We say that a set $X\subseteq\Re^n$ is \emph{closed under
negation} if $X=-X.$ Given a function
$f\colon X\to\Re,$ where $X\subseteq\Re^n$ is closed under
negation, we say that $f$ is
\emph{odd} (respectively, \emph{even}) if $f(-x)=-f(x)$ for all
$x\in X$ (respectively, $f(-x)=f(x)$ for all $x\in X$).

Given functions $f\colon X\to\moo$ and $\g\colon Y\to\moo,$ recall that the
function $f\wedge \g\colon  X\times Y\to\moo$ is given by $(f\wedge \g)(x,y)
= f(x)\wedge \g(y).$ The function $f\vee \g$ is defined analogously.
Observe that in this notation, \mbox{$f\wedge f$} and $f$ are completely
different functions, the former having domain $X\times X$ and the
latter $X.$ These conventions extend in the obvious way to any
number of functions. For example, $f_1\wedge f_2\wedge \cdots\wedge
f_k$ is a Boolean function with domain $X_1\times X_2\times\cdots\times
X_k,$ where $X_i$ is the domain of $f_i.$ 
Generalizing further, we let the symbol $F(f_1,\dots,f_k)$ 
denote the Boolean function on $X_1\times
X_2\times\cdots\times X_k$ obtained by composing a given function
$F\colon \mook\to\moo$ with the functions $f_1,f_2,\dots,f_k.$
Finally, recall that the negated function $\overline f\colon  X\to\moo$
is given by $\overline f(x)=-f(x).$

\subsection{Sign-representation and approximation by polynomials}

By the \emph{degree} of a multivariate polynomial $p$ on $\Re^n,$
denoted $\deg p,$ we shall always mean the total degree of $p,$
i.e., the greatest total degree of any monomial of $p.$ The
degree of a rational function $p(x)/q(x)$ is the maximum of $\deg
p$ and $\deg q.$
Given a function $f\colon X\to\moo,$ where $X\subset\Re^n$
is a finite set, the \emph{threshold degree} $\degthr(f)$ of
$f$ is defined as the least degree of a multivariate polynomial $p$
such that $f(x)p(x)>0$ for all $x\in X.$ In words, the threshold
degree of $f$ is the least degree of a polynomial that represents
$f$ in sign. Equivalent terms in the literature
include ``strong degree''~\cite{aspnes91voting}, ``voting polynomial
degree''~\cite{krause94depth2mod}, ``polynomial threshold
function degree''~\cite{OS-extremal-ptf}, and ``sign
degree''~\cite{buhrman07pp-upp}.  Crucial to understanding the
threshold degree is the following result, which is a well-known
corollary to Gordan's transposition theorem~\cite{gordan1873lp}.

\begin{theorem}[Gordan~\cite{gordan1873lp}]
\label{thm:gordan}
Let $X\subset\Re^n$ be a finite set, $f\colon X\to\moo$ a given
function.  Then $\degthr(f)>d$ if and only if there exists a
probability distribution $\mu$ on $X$ such that 
\begin{align*}
\sum_{x\in X}\mu(x)f(x)p(x)=0
\end{align*}
for every polynomial $p$ of degree up to $d.$
Equivalently, $\degthr(f)>d$ if and only if there exists a map
$\psi\colon X\to\Re,$ $\psi\not\equiv 0,$ such that $f(x)\psi(x)\geq
0$ on $X$ and
\begin{align*}
\sum_{x\in X}\psi(x)p(x)=0
\end{align*}
for every polynomial $p$ of degree up to $d.$
\end{theorem}

\noindent
Theorem~\ref{thm:gordan} has a short proof using linear
programming duality, as explained
in~\cite[\S2.2]{sherstov07ac-majmaj}.

The threshold degree is closely related to another analytic
notion. Let $f\colon X\to\moo$ be given, for a finite subset
$X\subset\Re^n.$ The \emph{$\epsilon$-approximate degree} of $f,$
denoted $\degeps(f),$ is the least degree of a polynomial $p$
such that $\abs{f(x)-p(x)}\leq \epsilon$ for all $x\in X.$  The
relationship between the threshold degree and approximate degree
is an obvious one:
\begin{align}
\degthr(f) = \lim_{\epsilon\nearrow1} \deg_{\epsilon}(f).
\label{eqn:adeg-thrdeg}
\end{align}
We will need the following dual characterization of the
approximate degree.

\begin{theorem}
\label{thm:dual-approx}
Fix $\epsilon\geq0.$ Let $f\colon X\to\moo$ be given, $X\subset\Re^n$ a
finite set. Then $\degeps(f)>d$ if and only if there exists a
function $\psi\colon X\to\Re$ such that
\begin{align*}
&\sum_{x\in X}|\psi(x)| =1, \\
&\sum_{x\in X}\psi(x)f(x) > \epsilon,\\
\intertext{and, for every polynomial $p$ of degree up to $d,$}
&\sum_{x\in X}\psi(x)p(x)=0.
\end{align*}
\end{theorem}

Theorem~\ref{thm:dual-approx} follows readily from linear programming
duality, as explained in~\cite[\S3]{sherstov07quantum}.
Theorem~\ref{thm:gordan} can be derived from Theorem~\ref{thm:dual-approx}
in view of (\ref{eqn:adeg-thrdeg}).

%
%
%
%

\subsection{Approximation by rational functions}

Consider a function $f\colon X\to\moo,$ where $X\subseteq \Re^n$
is an arbitrary set. For $d\geq0,$ we define
\[ R(f,d) \,= \,\inf_{\rule{0pt}{7pt}p,q} \,\sup_{x\in X} 
   \left\lvert f(x) - \frac{p(x)}{q(x)} \right\rvert,\]
where the infimum is over multivariate polynomials $p$ and $q$ of degree up
to $d$ such that $q$ does not vanish on $X.$
In words,
$R(f,d)$ is the least error in an approximation of $f$ by a
multivariate rational function of degree up to $d.$ 
We will also take an interest in the related quantity
\[ R^+(f,d) \,= \,\inf_{\rule{0pt}{7pt}p,q} \,\sup_{x\in X} 
   \left\lvert f(x) - \frac{p(x)}{q(x)} \right\rvert,\]
where the infimum is over multivariate polynomials $p$ and $q$ of
degree up to $d$ such that $q$ is positive on $X.$ These two
quantities are related in a straightforward way:
\begin{equation}
\label{eqn:rational-positive-denominator}
 R^+(f,2d) \leq R(f,d) \leq R^+(f,d). 
\end{equation}
The second inequality here is trivial. The first follows from the fact that
every rational approximant $p(x)/q(x)$ of degree $d$ gives rise to a
degree-$2d$ rational approximant with the same error and a positive
denominator, namely, $\{p(x)q(x)\}/q(x)^2.$
The infimum in the definitions of $R(f,d)$ and $R^+(f,d)$ cannot
in general be replaced by a minimum~\cite{rivlin-book}, even when
$X$ is a finite subset of $\Re.$ This is in contrast to the more
familiar setting of a finite-dimensional normed linear space,
where least-error approximants are guaranteed to exist.
We now recall Newman's classical construction of a rational
approximant to the sign function~\cite{newman64rational}.

\begin{theorem}[Newman]
\label{thm:newman-approx}
Fix $N>1.$ Then for every integer $k\geq1,$ there is a rational
function $S(t)$ of degree $k$ such that
\begin{align}
 \max_{1\leq|t|\leq N} |\sign t - S(t)|\leq 1-N^{-1/k} 
 \label{eqn:newman-approx}
\end{align}
and the denominator of $S$ is positive on $[-N,-1]\cup[1,N].$
\end{theorem}

\begin{proof}
[Proof \textup{(adapted from Newman~\cite{newman64rational})}]
Consider the univariate polynomial
\[p(t) = \prod_{i=1}^k \big(t+N^{(2i-1)/(2k)}\big).\] 
By examining every interval $[N^{i/(2k)},N^{(i+1)/(2k)}],$ where
$i=0,1,\dots,2k-1,$ one sees that
\begin{align}
p(t) \geq\frac{N^{1/(2k)}+1}{N^{1/(2k)}-1}\, \lvert p(-t)\rvert,
   \qquad 1\leq t\leq N. 
   \label{eqn:newman-balance}
\end{align}
Letting 
\begin{align*}
S(t)=N^{-1/(2k)}\cdot \frac{p(t)-p(-t)}{p(t)+p(-t)},
\end{align*}
one has (\ref{eqn:newman-approx}).  The positivity of the
denominator of $S$ on $[-N,-1]\cup[1,N]$ is a consequence of
(\ref{eqn:newman-balance}).
\end{proof}

A useful consequence of Newman's theorem is the following general
statement on decreasing the error in rational approximation.

\begin{theorem}
\label{thm:error-boosting}
Let $f\colon X\to\moo$ be given, where $X\subseteq\Re^n.$ Let $d$ be a given integer,
$\epsilon=R(f,d).$ Then for $k=1,2,3,\dots,$
\begin{align*}
R(f,kd) \leq 1 - \PARENS{\frac{1-\epsilon}{1+\epsilon}}^{1/k}.
\end{align*}
\end{theorem}

\begin{proof}
We may assume that $\epsilon<1,$ the theorem being trivial
otherwise.  Let $S$ be the degree-$k$ rational approximant to the
sign function for $N=(1+\epsilon)/(1-\epsilon),$ as constructed
in Theorem~\ref{thm:newman-approx}.  Let
$A_1,A_2,\dots,A_m,\dots$ be a sequence of rational functions on
$X$ of degree at most $d$ such that $\sup_X|f-A_m|\to\epsilon$ as
$m\to\infty.$ The theorem follows by considering the sequence of
approximants $S(A_m(x)/\{1-\epsilon\})$ as $m\to\infty.$
\end{proof}

\subsection{Symmetrization}

Let $S_n$ denote the symmetric group on $n$ elements. For $\sigma\in
S_n$ and $x\in\Re^n$, we denote $\sigma
x=(x_{\sigma(1)},\ldots,x_{\sigma(n)})\in\Re^n.$ The following is
a generalized form of Minsky and Papert's \emph{symmetrization
argument}~\cite{minsky88perceptrons}, as formulated in~\cite{RS07dc-dnf}.

\begin{proposition}[cf.~Minsky and Papert]
\label{prop:symm-argument}
Let $n_1,\dots,n_k$ be positive integers.  Let
$\phi\colon\zoo^{n_1}\times\cdots\times\zoo^{n_k}\to\Re$ be a
polynomial of degree $d.$ Then there is a polynomial $p$ on $\Re^k$
of degree at most $d$ such that for all $x$ in the domain of $\phi,$ 
\begin{align*}
\Exp_{\sigma_1\in S_{n_1},\dots,\sigma_k\in S_{n_k}}
\left[\phi\big(\sigma_1x_1,\dots,\sigma_kx_k\big)\right] =
p\big(\dots,x_{i,1}+\cdots+x_{i,n_i},\dots\big).
\end{align*}
\end{proposition}

We now obtain a form of the symmetrization argument for
rational approximation.

\begin{proposition}
\label{prop:symm-rational}
Let $n_1,\dots,n_k$ be positive integers, and $\alpha,\beta$ distinct
reals. Let
$G\colon\{\alpha,\beta\}^{n_1}\times\cdots\times\{\alpha,\beta\}^{n_k}\to\moo$
be a function such that $G(x_1,\dots,x_k)\equiv G(\sigma_1
x_1,\dots,\sigma_k x_k)$ for all $\sigma_1\in S_{n_1},\dots,\sigma_k\in
S_{n_k}.$ Let $d$ be a given integer. Then for each $\epsilon>R^+(G,d),$
there exists a rational function $p/q$ on $\Re^k$ of degree at most
$d$ such that for all $x$ in the domain of $G,$ one has
\begin{align*}
\left\lvert G(x)-\frac
{p(\dots,x_{i,1}+\cdots+x_{i,n_i},\dots)}
{q(\dots,x_{i,1}+\cdots+x_{i,n_i},\dots)}
\right\rvert<\epsilon
\end{align*}
and $q(\dots,x_{i,1}+\cdots+x_{i,n_i},\dots)>0.$ 
\end{proposition}

\begin{proof}
Clearly, we may assume that $\epsilon<1.$
Using the linear bijection $(\alpha,\beta)\leftrightarrow (0,1)$ if
necessary, we may further assume that $\alpha=0$ and $\beta=1.$
Since $\epsilon>R^+(G,d),$ there are polynomials $P,Q$ 
of degree up to $d$ such that for all $x$ in the domain of $G,$ one
has $Q(x)>0$ and
\begin{align*}
(1-\epsilon)Q(x)<G(x)P(x)<(1+\epsilon)Q(x).
\end{align*}
By Proposition~\ref{prop:symm-argument}, there exist polynomials
$p,q$ on $\Re^k$ of degree at most $d$ such that
\begin{align*}
\Exp_{\sigma_1\in S_{n_1},\ldots,\sigma_k\in S_{n_k}}
\left[P\big(\sigma_1x_1,\dots,\sigma_kx_k\big)\right] =
p\big(\dots,x_{i,1}+\cdots+x_{i,n_i},\dots\big)
\end{align*}
and
\begin{align*}
\Exp_{\sigma_1\in S_{n_1},\ldots,\sigma_k\in S_{n_k}}
\left[Q\big(\sigma_1x_1,\dots,\sigma_kx_k\big)\right] =
q\big(\dots,x_{i,1}+\cdots+x_{i,n_i},\dots\big)
\end{align*}
for all $x$ in the domain of $G.$ Then the required properties of 
$p$ and $q$ follow immediately from the corresponding properties of $P$ and
$Q.$
\end{proof}

\section{Direct product theorems}

In the several subsections that follow, we prove our direct
product theorems for polynomial representations of
composed Boolean functions. General compositions are treated in
Section~\ref{sec:dp}, followed by a study of conjunctions and
other specific compositions in
Sections~\ref{sec:auxiliary}--\ref{sec:additional-conj}.

\subsection{General compositions}
\label{sec:dp}

We begin our study with general
compositions of the form $F(f_1,\dots,f_k).$ Our focus in this
section will be on results that depend only on the threshold or
approximate degrees of $F,f_1,\dots,f_k.$ In later sections, we will
exploit additional structure of the functions involved. The following result
settles Theorems~\ref{thm:thrdeg-dp} and~\ref{thm:main-approx-dp}
from the Introduction.

\begin{theorem}
\label{thm:thrdeg-adeg-dp}
Let $f\colon X\to\moo$ and $F\colon \mook\to\moo$ be given functions, where
$X\subset\Re^n$ is a finite set. Then for $0<\epsilon<1,$
\begin{align}
\deg_\eps(F(f,\dots,f)) \geq \deg_\eps(F)\degthr(f).
\label{eqn:adeg-dp}
\end{align}
In particular,
\begin{align}
\degthr(F(f,\dots,f)) \geq \degthr(F)\degthr(f).
\label{eqn:thrdeg-dp}
\end{align}
\end{theorem}

\begin{proof}
Recall that the threshold degree is a limiting case of the
approximate degree, as given by (\ref{eqn:adeg-thrdeg}).  
Hence, one obtains (\ref{eqn:thrdeg-dp}) by letting $\epsilon\nearrow1$
in (\ref{eqn:adeg-dp}).  In the remainder of the proof, we focus
on (\ref{eqn:adeg-dp}) alone.

Put $D=\degeps(F)$ and $d=\degthr(f).$ By
Theorem~\ref{thm:dual-approx}, there
exists a map $\Psi\colon \mook\to\Re$ such that
\begin{align}
&\sum_{z\in\mook}\abs{\Psi(z)} =1, \label{eqn:Psi-bounded}\\
&\sum_{z\in\mook}\Psi(z)F(z) > \epsilon \label{eqn:Psi-correl},
\end{align}
and  $\sum\Psi(z)p(z)=0$ for every polynomial $p$ of
degree less than $D.$ By Theorem~\ref{thm:gordan}, 
there exists a distribution $\mu$ on $X$ such that
$\sum\mu(x)f(x)p(x)=0$ for every polynomial $p$ of
degree less than $d.$

Now, define $\zeta\colon X^k\to\Re$ by 
\begin{align*}
\zeta(\dots,x_i,\dots)
= 2^k\Psi(\dots,f(x_i),\dots) \prod_{i=1}^k \mu(x_i).
\end{align*}
We claim that 
\begin{align}
\label{eqn:zeta-orthogonality}
\sum_{X^k} \zeta(\dots,x_i,\dots)p(\dots,x_i,\dots) = 0
\end{align}
for every polynomial $p$ of degree less than $Dd.$ By linearity,
it suffices to consider a polynomial $p$ of the form $p(\dots,x_i,\dots)
= \prod p_i(x_i),$ where $\sum \deg p_i<Dd.$ Since $\Psi$ is
orthogonal on $\mook$ to all polynomials of degree less than $D,$
we have the representation
\begin{align*}
\Psi(z) = \sum_{\substack{S\subseteq\{1,\dots,k\},\\|S|\geq D}} \hat\Psi(S)
\prod_{i\in S} z_i
\end{align*}
for some reals $\hat\Psi(S).$ As a result,
\begin{multline}
\sum_{X^k} \zeta(\dots,x_i,\dots)p(\dots,x_i,\dots) \\
=2^k \sum_{|S|\geq D} \hat\Psi(S)
\prod_{i\in S} \undercbrace{\PARENS{\sum_{x_i\in
X}^{\phantom{A}} \mu(x_i)f(x_i)p_i(x_i)}}
		 \prod_{i\notin S}\PARENS{\sum_{x_i\in X}^{\phantom{A}}
		 \mu(x_i)p_i(x_i)}.
		 \label{eqn:big-sum}
\end{multline}
Since $\sum \deg p_i<Dd,$ the pigeonhole principle implies that
$\deg p_i<d$ for more than $k-D$ indices $i\in\{1,\dots,k\}.$ As a
result, for each set $S$ in the outer summation of (\ref{eqn:big-sum}),
at least one of the underbraced factors vanishes (recall that
$f$ is orthogonal on $X$ with respect to $\mu$
to all polynomials of degree less than
$d$). This gives (\ref{eqn:zeta-orthogonality}).

We may assume that $f$
is not a constant function, the theorem being trivial otherwise.
It follows that $\degthr(f)\geq1$ and $\sum_X \mu(x)f(x)=0.$ 
Now, define a product distribution $\lambda$ on $X^k$ by
$\lambda(\dots,x_i,\dots)=\prod \mu(x_i).$ Since $\sum_X
\mu(x)f(x)=0,$ it follows that the string
$(\dots,f(x_i),\dots)$ is distributed uniformly on $\mook$
when $(\dots,x_i,\dots)\sim\lambda.$ As a result,
\begin{align}
\sum_{X^k} |\zeta(\dots,x_i,\dots)|
= 2^k \Exp_{z\in\mook}[\abs{\Psi(\dots,z_i,\dots)}]=1,
\label{eqn:zeta-bounded}
\end{align}
where the last equality holds by (\ref{eqn:Psi-bounded}).
Similarly,
\begin{multline}
\sum_{X^k} \zeta(\dots,x_i,\dots)F(\dots,f(x_i),\dots)\\
= 2^k
\Exp_{z\in\mook}[\Psi(\dots,z_i,\dots)F(\dots,z_i,\dots)]>\epsilon,
\qquad
\label{eqn:zeta-correl}
\end{multline}
where the inequality holds by (\ref{eqn:Psi-correl}).  
Now (\ref{eqn:adeg-dp}) follows from
(\ref{eqn:zeta-orthogonality}), (\ref{eqn:zeta-bounded}), 
(\ref{eqn:zeta-correl}), and Theorem~\ref{thm:dual-approx}.
\end{proof}

\begin{remark*}
In Theorem~\ref{thm:thrdeg-adeg-dp} and elsewhere in this paper,
we consider Boolean functions on finite subsets of $\Re^n,$ which
is the setting of primary interest in computational complexity.
It is useful to keep in mind, however, that approximation and
sign-representation problems on compact infinite sets and other
well-behaved infinite sets are easily reduced to the finite case.
\end{remark*}

We now
consider the so-called AND-OR tree, given by
$f(x)=\bigvee_{i=1}^n \bigwedge_{j=1}^n x_{ij}.$ We improve the
standing lower bound on the approximate degree of $f$ from
$\Omega(n^{0.66\dots})$ to $\Omega(n^{0.75}),$ the best upper
bound being $O(n).$

\bigskip

\begin{restatetheorem}{thm:main-and-or}{\textsc{restated}}
Let $f\colon \moo^{n^2}\to\moo$ be given by
$
f(x) = \bigvee_{i=1}^n \bigwedge_{j=1}^n x_{ij}.
$
Then 
\begin{align*}
\deg_{1/3}(f) = \Omega(n^{0.75}).
\end{align*}
\end{restatetheorem}

\begin{proof}
Without loss of generality, assume that $n=4m^2$ for some integer $m.$  
Define $\g\colon\moo^{4m^3}\to\moo$ by 
\begin{align*}
\g(x) = \bigvee_{i=1}^m \bigwedge_{j=1}^{4m^2} x_{ij}.
\end{align*}
Let $G\colon\moo^{4m}\to\moo$ be given by $G(x)=x_1\vee\cdots\vee
x_{4m}.$
A well-known result of Minsky and Papert~\cite{minsky88perceptrons}
states that $\degthr(\g)=m.$ Also, Nisan and
Szegedy~\cite{nisan-szegedy94degree} proved that
$\deg_{1/3}(G)=\Theta(\sqrt m).$ Since $f=G(\g,\dots,\g),$ it follows
by Theorem~\ref{thm:thrdeg-adeg-dp} that $\deg_{1/3}(f)=\Omega(m\sqrt
m),$ as desired.
\end{proof}

We now further develop the ideas of
Theorem~\ref{thm:thrdeg-adeg-dp} to
obtain a more general result on the approximation of composed
functions by polynomials. This generalization is based on a
combinatorial property of Boolean functions known as
\emph{certificate complexity}.  For a string $x\in\mook$ and a
set $S\subseteq\onetok$ whose distinct elements are
$i_1<i_2<\cdots<i_{|S|},$ we adopt the notation $x|_S =
(x_{i_1},x_{i_2},\dots,x_{i_{|S|}})\in\zoo^{|S|}.$ For a Boolean
function $F\colon\mook\to\moo$ and a point $x\in\mook,$ the
\emph{certificate complexity of $F$ at $x,$} denoted $C_x(F),$ is
the minimum size of a subset $S\subseteq\onetok$ such that
$F(x)=F(y)$ for all $y\in\mook$ with $x|_S=y|_S.$ The
\emph{certificate complexity of $F,$} denoted $C(F),$ is the
maximum $C_x(F)$ over all $x.$ In the degenerate case when $F$ is
constant, we have $C(F)=0.$ At the other extreme, the parity
function $F\colon\mook\to\moo$ satisfies $C(F)=k,$ which is the
maximum possible.  The following proposition is immediate from
the definition of certificate complexity.

\begin{proposition}
\label{prop:certificate}
Let $F\colon\mook\to\moo$ be a given Boolean function. Let
$y\in\mook$ be a random string whose $i$th bit is
set to $-1$ with probability $\alpha_i$ and to $+1$ otherwise,
independently for each $i.$ 
Then for every $x\in\mook,$
\begin{align*}
\Prob_{y}[F(x_1,\dots,x_k)= F(x_1y_1,\dots,x_ky_k)] \geq
\min_{i_1<i_2<\cdots<i_{C_x(F)}}
\prod_{j=1}^{C_x(F)}(1-\alpha_{i_j}).
\end{align*}
\end{proposition}

\begin{proof}
Fix a set $S\subseteq\onetok$ of cardinality $C_x(F)$ such that
$F(x)=F(y)$ whenever $x|_S=y|_S.$ Then clearly
$\Prob_{y}[F(\dots,x_i,\dots)=F(\dots,x_iy_i,\dots)] \geq
\Prob_{y}[y|_S=(1,1,\dots,1)],$ and the bound follows.
\end{proof}

We can now state and prove the desired generalization of
Theorem~\ref{thm:thrdeg-adeg-dp}.

\begin{theorem}
   \label{thm:approx-dp}
Let $f\colon X\to\moo$ and $F\colon \mook\to\moo$ be given functions, where
$X\subset\Re^n$ is a finite set. Then for each $\epsilon,\delta>0,$
\begin{align}
\deg_{\epsilon+\eta-2+2(1-\delta)^{C(F)}}(F(f,\dots,f)) &\geq \deg_{\epsilon}(F)
\deg_{1-\delta}(f)
\label{eqn:approx-dp}
\end{align}
for some $\eta=\eta(\epsilon,F)>0.$ 
\end{theorem}

\begin{remark}
One recovers Theorem~\ref{thm:thrdeg-adeg-dp} 
by letting $\delta\searrow0$ in (\ref{eqn:approx-dp}). 
We also note that (\ref{eqn:approx-dp}) is
considerably stronger than Theorem~\ref{thm:thrdeg-adeg-dp}: functions
$\mook\to\moo$ are known, such as {\sc odd-max-bit}~\cite{beigel94perceptrons},
with threshold degree~$1$ and \mbox{$(1-\delta)$}-approximate degree
$k^{\Omega(1)}$ for $\delta$ as small as
$\delta=\exp\{-k^{\Omega(1)}\}.$ Another advantage of
Theorem~\ref{thm:approx-dp} is that the $(1-\delta)$-approximate
degree is easier to bound from below than the threshold
degree~\cite{beigel94perceptrons, vereshchagin95weight, mlj07sq,
podolskii07perceptrons, podolskii08perceptrons}, even for
$\delta$ exponentially small. For $\delta$ small, the
$(1-\delta)$-approximate degree is essentially equivalent to a
notion known as \emph{perceptron
weight}~\cite{minsky88perceptrons, beigel94perceptrons,
vereshchagin95weight, KP98threshold, KOS:02,
klivans-servedio06decision-lists, mlj07sq, buhrman07pp-upp,
podolskii07perceptrons, podolskii08perceptrons}.
\end{remark}

\begin{proof}[Proof of Theorem~\textup{\ref{thm:approx-dp}}]
Let $D=\degeps(F)$ and $d=\deg_{1-\delta}(f)>0.$ 
Theorem~\ref{thm:dual-approx} provides 
a map $\Psi\colon \mook\to\Re$ such that
\begin{align}
&\sum_{z\in\mook}\abs{\Psi(z)} =1, \label{eqn:Psi-bounded2}\\
&\sum_{z\in\mook}\Psi(z)F(z) > \epsilon+\eta
\label{eqn:Psi-correl2}
\end{align}
for some $\eta=\eta(\epsilon,F)>0,$
and  $\sum_{z\in\mook}\Psi(z)p(z)=0$ for every polynomial $p$ of
degree less than $D.$ Analogously, there
exists a map $\psi\colon X\to\Re$ such that
\begin{align}
&\sum_{x\in X}|\psi(x)| =1, \label{eqn:psi-bounded2}\\
&\sum_{x\in X}\psi(x)f(x) > 1-\delta, \label{eqn:psi-correl2}
\end{align}
and $\sum_{x\in X}\psi(x)p(x)=0$ for every polynomial $p$ of degree
less than $d.$

Define $\zeta\colon X^k\to\Re$ by 
\begin{align*}
\zeta(\dots,x_i,\dots)
= 2^k\,\Psi(\dots,\Sign \psi(x_i),\dots) \prod_{i=1}^k \abs{\psi(x_i)}.
\end{align*}
By the same argument as in Theorem~\ref{thm:thrdeg-adeg-dp}, we have 
\begin{align}
\sum_{X^k} \zeta(\dots,x_i,\dots)p(\dots,x_i,\dots) = 0
\label{eqn:zeta-orthogonality2}
\end{align}
for every polynomial $p$ of degree less than $Dd.$ 

Let $\mu$ be the distribution on $X^k$ given by $\mu(\dots,x_i,\dots)=\prod
\, \abs{\psi(x_i)}.$ Since $\psi$ is orthogonal to the constant
polynomial $1,$ the string $(\dots,\Sign\psi(x_i),\dots)$ is
distributed uniformly over $\mook$ when one samples
$(\dots,x_i,\dots)$ according to $\mu.$ As a result,
\begin{align}
\sum_{X^k}\,  \abs{\zeta(\dots,x_i,\dots)}
 = \sum_{z\in\mook}\abs{\Psi(z)}
 = 1,
 \label{eqn:zeta-bounded2}
\end{align}
where the final equality uses (\ref{eqn:Psi-bounded2}).

Define 
$A_{+1}=\{x\in X:\psi(x)>0, f(x)=-1\}$ and 
$A_{-1}=\{x\in X:\psi(x)<0, f(x)=+1\}.$  Since $\psi$ is orthogonal to the
constant polynomial~$1,$ it follows from (\ref{eqn:psi-bounded2}) that 
\begin{align*}
\sum_{x:\psi(x)<0} \abs{\psi(x)} = 
\sum_{x:\psi(x)>0} \abs{\psi(x)} = \frac12.
\end{align*}
In light of (\ref{eqn:psi-correl2}), we see that
$\sum_{x\in A_{+1}} \abs{\psi(x)} < \delta/2$ and 
$\sum_{x\in A_{-1}} \abs{\psi(x)} < \delta/2.$ Now, for any given
$z\in\mook,$ the following two random variables are identically
distributed: 

\begin{itemize}
\item the string $(\dots,f(x_i),\dots)$ when one chooses 
$(\dots,x_i,\dots)\sim\mu$ and conditions on the event
that $(\dots,\Sign\psi(x_i),\dots)=z$;

\item the string $(\dots,y_iz_i,\dots),$ where $y\in\mook$ is a
random string whose $i$th bit independently takes on $-1$
with probability $2\sum_{x\in A_{z_i}}\abs{\psi(x)}<\delta.$
\end{itemize}
Proposition~\ref{prop:certificate}
now implies that for each $z\in\mook,$
\begin{multline}
\left\lvert\Exp_{\mu}\Big[F(\dots,f(x_i),\dots) \mid (\dots,\Sign
\psi(x_i),\dots)=z\Big]\right.\\
\left.\phantom{\Exp_{\mu}}-F(\dots,\Sign\psi(x_i),\dots)\right\rvert 
\leq 2-2(1-\delta)^{C(F)}.\qquad
\label{eqn:error-prob}
\end{multline}
We are now prepared to complete the proof. We have
\begin{align}
\sum_{X^k} &\zeta(\dots,x_i\dots)F(\dots,f(x_i),\dots)
\nonumber\\
&=
2^k\Exp_{\mu}\Big[\Psi(\dots,\Sign\psi(x_i),\dots)F(\dots,f(x_i),\dots)\Big]
\nonumber\\
&\geq \sum_{z\in\mook}\Psi(z)F(z) - 2\{1-(1-\delta)^{C(F)}\}
\sum_{z\in\mook}\abs{\Psi(z)} 
&&\nonumber\\
&> \epsilon +\eta- 2 + 2(1-\delta)^{C(F)},
\label{eqn:zeta-correlated2}
\end{align}
where the last two inequalities use (\ref{eqn:error-prob}),
(\ref{eqn:Psi-bounded2}), and (\ref{eqn:Psi-correl2}). In
view of Theorem~\ref{thm:dual-approx}, the exhibited properties
(\ref{eqn:zeta-orthogonality2}), (\ref{eqn:zeta-bounded2}), and
(\ref{eqn:zeta-correlated2}) of $\zeta$ force (\ref{eqn:approx-dp}).
\end{proof}

Theorems~\ref{thm:thrdeg-adeg-dp}
and~\ref{thm:approx-dp} complement known \emph{upper} bounds for
the approximation of composed functions. The following theorem is
due to Buhrman et al.~\cite{BNRW05robust}, who
studied the approximation of Boolean functions with perturbed
inputs. We include the proof from~\cite{BNRW05robust} and
slightly generalize it to any given parameters.

\begin{theorem}[cf. Buhrman et al.]
\label{thm:approx-upper-F-f}
Fix functions $F\colon\mook\to\moo$ and $f\colon X\to\moo,$ where
$X\subset\Re^n$ is finite.  Then for all $\Delta,\delta\geq 0,$
\begin{align}
\deg_{\eta(\Delta,\delta)}(F(f,\dots,f)) 
  \leq \deg_\Delta(F)\deg_\delta(f),
  \label{eqn:Delta-delta}
\end{align}
where 
\begin{align}
\eta(\Delta,\delta) = 
\Delta + 2-2\left(1-\frac{\delta}{1+\delta}\right)^{C(F)}.
\label{eqn:approx-upper-F-f-nu}
\end{align}
In particular,
\begin{multline}
\deg_{1/3}(F(f,\dots,f)) \\\leq \deg_{1/3}(F)\deg_{1/3}(f)\cdot O(\log
\{1+\deg_{1/3}(F)\}).\qquad
  \label{eqn:approx-upper}
\end{multline}
\end{theorem}

\begin{proof}[Proof \textup{(adapted from Buhrman et al.)}.]
Fix polynomials $P$ and $p$ on $\mook$ and $X,$
respectively. As usual, $P$ may be assumed to be multilinear in
view of its domain. 
Define $\Phi\colon X^k\to\Re$ by 
\begin{align*}
\Phi(\dots,x_i,\dots) = P\PARENS{\dots,
\frac1{1+\|f-p\|_\infty}p(x_i),\dots}.
\end{align*}
Fix any input $(\dots,x_i,\dots)\in X^k$ and
consider a random variable $y\in\mook$ whose $i$th bit takes on
$-1$ with probability 
\begin{align*}
\alpha_i = \frac12 - \frac{f(x_i)p(x_i)}{2(1+\|f-p\|_\infty)}
\leq \frac{\|f-p\|_\infty}{1+\|f-p\|_\infty},
\end{align*}
independently for each $i.$ Then
\begin{align*}
\abs{\Phi(\dots,x_i,\dots) &- F(\dots,f(x_i),\dots)}\\
   &= \left|\Exp_y[P(\dots,y_if(x_i),\dots) -
   F(\dots,f(x_i),\dots)]\right|\\
   &\leq \|P-F\|_\infty + 
       \left|\Exp_y[F(\dots,y_if(x_i),\dots) -
       F(\dots,f(x_i),\dots)]\right|\\
   & \leq \|P-F\|_\infty + 
      2 -
	  2\PARENS{1-\frac{\|f-p\|_\infty}{1+\|f-p\|_\infty}}^{C(F)},
\end{align*}
where the first and last steps in the derivation follow by the
multilinearity of $P$ and by Proposition~\ref{prop:certificate},
respectively.  This completes the proof of
(\ref{eqn:Delta-delta}).

Taking $\Delta=1/6$ and $\delta=1/(12C(F))$ in 
(\ref{eqn:Delta-delta}) gives
\begin{align*}
\deg_{1/3}(F(f,\dots,f)) \leq \deg_{1/6}(F)\deg_{1/(12C(F))}(f).
\end{align*}
Basic approximation theory~\cite{eremenko07sign} shows
that for each $\epsilon>0,$ there exists a univariate polynomial
of degree $O(\log \frac1\epsilon)$ that sends $[-\frac43,-\frac23]\to
[-1-\epsilon,-1+\epsilon]$ and $[\frac23,\frac43] \to
[1-\epsilon,1+\epsilon].$ As a result, we obtain 
\begin{align*}
\deg_{1/3}(F(f,\dots,f)) \leq
\deg_{1/3}(F)\deg_{1/3}(f)\cdot O(\log\{1+C(F)\}),
\end{align*}
which is equivalent to (\ref{eqn:approx-upper}) because $C(F)$ is
known to be within a polynomial of $\deg_{1/3}(F)$ for every
Boolean function $F\colon\mook\to\moo,$ as discussed in detail in
the survey article~\cite{buhrman-dewolf02DT-survey}.
\end{proof}

\paragraph{Compositions with \emph{k} distinct functions.}
We now consider compositions of the form $F(f_1,\dots,f_k),$ where the
functions $f_1,\dots,f_k$ may all be distinct.
For a function $F\colon \mook\to\Re$ and a vector
$v=(v_1,\dots,v_k)$ of nonnegative integers, define the
\emph{$(\epsilon,v)$-approximate degree} $\deg_{\epsilon,v}(F)$ to be
the least $D$ for which there is a polynomial $P(x_1,\dots,x_k)$
with 
\begin{align*}
P\in\Span\BRACES{\prod_{i\in S} x_i 
:S\subseteq\onetok, \; \sum_{i\in S}^{\phantom{S}} v_i\leq D}
\end{align*}
and $\|F-P\|_\infty\leq\epsilon.$ Note that the
$\epsilon$-approximate degree of $F$ is the
$(\epsilon,v)$-approximate degree of $F$ for $v=(1,1,\dots,1).$
It is clear that
\begin{align*}
\deg_{\eps,v}(F) \geq \min_{i_1<i_2<\cdots<i_{\deg_{\eps}(F)}}
\{v_{i_1}+v_{i_2}+\cdots+v_{i_{\deg_{\eps}(F)}}\},
\end{align*}
with an arbitrary gap achievable between the right and left
members of the inequality. 
We will also need the following generalized version of
Theorem~\ref{thm:dual-approx}, due to
Ioffe and Tikhomirov~\cite{ioffe-tikhomirov68duality}.

\begin{theorem}[Ioffe and Tikhomirov]
\label{thm:berr-morth}
Let $X$ be a finite set. Fix any family $\Phi$ of functions
$X\to\Re$ and an additional function
$f\colon X\to\Re.$ Then
\begin{align*}
 \min_{\phi\in\Span(\Phi)} \|f - \phi\|_\infty = 
  \max_\psi \left\{ \sum_{x\in X} f(x)\psi(x) \right\},
\end{align*}
where the maximum is over all functions $\psi\colon X\to\Re$ such that
\begin{align*}
\sum_{x\in X} |\psi(x)| \leq 1
\end{align*}
and, for each $\phi\in\Phi,$
\begin{align*}
\sum_{x\in X}\phi(x)\psi(x)=0.
\end{align*}
\end{theorem}

\noindent
A short proof of Theorem~\ref{thm:berr-morth} 
can be found, e.g., in~\cite[\S3]{sherstov07quantum}.
With this setup in place, we obtain the following analogues
of Theorems~\ref{thm:approx-dp} and~\ref{thm:approx-upper-F-f}
for compositions of the form $F(f_1,\dots,f_k).$

\begin{theorem}
   \label{thm:approx-dp-generalized}
Fix nonconstant functions $F\colon\mook\to\moo$ and $f_i\colon X_i\to\moo,$
$i=1,2,\dots,k,$ where each $X_i\subset\Re^n$ is finite.
Then for $\epsilon,\delta>0,$ one has
\begin{align}
\deg_{\epsilon+\eta-2+2(1-\delta)^{C(F)}}(F(f_1,\dots,f_k)) 
&\geq \deg_{\epsilon,v}(F)
\label{eqn:approx-dp-generalized}
\end{align}
for some $\eta=\eta(\epsilon,F)>0,$ 
where $v=(\deg_{1-\delta}(f_1),\dots,\deg_{1-\delta}(f_k)).$
\end{theorem}

\begin{proof}
Let $D=\deg_{\eps,v}(F)$ and $d_i=\deg_{1-\delta}(f_i).$
Theorem~\ref{thm:berr-morth} provides 
a map $\Psi\colon \mook\to\Re$ such that
\begin{align}
&\sum_{z\in\mook}\abs{\Psi(z)} =1, \label{eqn:Psi-bounded3}\\
&\sum_{z\in\mook}\Psi(z)F(z) > \epsilon+\eta
\nonumber
\end{align}
for some $\eta=\eta(\epsilon,F)>0,$
and  
\begin{align*}
\Psi(z) = \sum_{S\in\mathcal{S}} \hat\Psi(S)\prod_{i\in S} z_i
\end{align*}
for some reals $\hat\Psi(S),$ where $\mathcal{S} =
\{S\subseteq\onetok: \sum_{i\in S}d_i\geq D\}.$
Analogously, there are maps $\psi_i\colon X_i\to\Re,$
$i=1,2,\dots,k,$ such that
\begin{align*}
&\sum_{x_i\in X_i}|\psi_i(x_i)| =1,               \\
&\sum_{x_i\in X_i}\psi_i(x_i)f_i(x_i) > 1-\delta, 
\end{align*}
and $\sum_{x_i\in X_i}\psi_i(x_i)p(x_i)=0$ for every polynomial $p$ of degree
less than $d_i.$

Define $\zeta\colon X_1\times\cdots\times X_k\to\Re$ by 
\begin{align*}
\zeta(\dots,x_i,\dots)
= 2^k\,\Psi(\dots,\Sign \psi_i(x_i),\dots) \prod_{i=1}^k \abs{\psi_i(x_i)}.
\end{align*}
By an argument analogous to that in 
Theorem~\ref{thm:thrdeg-adeg-dp}, we have 
\begin{align}
\sum_{X_1\times\cdots\times X_k} \zeta(\dots,x_i,\dots)p(\dots,x_i,\dots) = 0
\label{eqn:zeta-orthogonality3}
\end{align}
for every polynomial $p$ of degree less than $D.$ 

Let $\mu$ be the distribution on $X_1\times\cdots\times X_k$ 
given by $\mu(\dots,x_i,\dots)=\prod
\, \abs{\psi_i(x_i)}.$ Since each $\psi_i$ is orthogonal to the constant
polynomial $1,$ the string $(\dots,\Sign\psi_i(x_i),\dots)$ is
distributed uniformly over $\mook$ when one samples
$(\dots,x_i,\dots)$ according to $\mu.$ As a result,
\begin{align}
\sum_{X_1\times\cdots\times X_k}\,  \abs{\zeta(\dots,x_i,\dots)}
 = \sum_{z\in\mook}\abs{\Psi(z)}
 = 1,
 \label{eqn:zeta-bounded3}
\end{align}
where the final equality uses (\ref{eqn:Psi-bounded3}).

By an argument analogous to that in Theorem~\ref{thm:approx-dp}, we obtain 
\begin{align}
\sum_{X_1\times\cdots\times X_k} &\zeta(\dots,x_i\dots)F(\dots,f_i(x_i),\dots)
> \epsilon +\eta- 2 + 2(1-\delta)^{C(F)}.
\label{eqn:zeta-correlated3}
\end{align}
In view of Theorem~\ref{thm:dual-approx}, the exhibited properties
(\ref{eqn:zeta-orthogonality3}), (\ref{eqn:zeta-bounded3}), and
(\ref{eqn:zeta-correlated3}) of $\zeta$ complete the proof.
\end{proof}

\begin{remark}
Analogous to the earlier development, 
taking $\delta\searrow0$ in
Theorem~\ref{thm:approx-dp-generalized} yields the lower bound
$
\deg_{\epsilon}(F(f_1,\dots,f_k)) 
\geq \deg_{\epsilon,v}(F)
$
for each $\epsilon>0,$ where 
$v=(\degthr(f_1),\dots,\degthr(f_k)).$
\end{remark}

\begin{theorem}
\label{thm:approx-upper-F-f-generalized}
Fix functions $F\colon\mook\to\moo$ and $f_i\colon X_i\to\moo,$
$i=1,2,\dots,k,$ where each $X_i\subset\Re^n$ is finite.  Then
for all $\Delta,\delta\geq 0,$
\begin{align*}
\deg_{\eta(\Delta,\delta)}(F(f_1,\dots,f_k)) 
  \leq \deg_{\Delta,v}(F),
\end{align*}
where 
$v = (\deg_\delta(f_1),\dots,\deg_\delta(f_k))$ and
\begin{align}
\eta(\Delta,\delta) &= 
\Delta + 2-2\left(1-\frac{\delta}{1+\delta}\right)^{C(F)}.
\label{eqn:approx-upper-F-f-generalized}
\end{align}
In particular,
\begin{align}
\deg_{1/3}(F(f_1,\dots,f_k)) = \deg_{{1/3},v}(F)\cdot O(\log
\{1+\deg_{1/3}(F)\})
  \label{eqn:approx-upper-generalized}
\end{align}
for $v=(\deg_{1/3}(f_1),\dots,\deg_{1/3}(f_k)).$
\end{theorem}

\begin{proof}
Fix a real polynomial $P$ on $\mook$ and polynomials 
$p_i$ on $X_i,$ respectively.
As usual, $P$ may be assumed to be multilinear in
view of its domain. 
Define $\Phi\colon X_1\times\cdots\times X_k\to\Re$ by 
\begin{align*}
\Phi(\dots,x_i,\dots) = P\PARENS{\dots,
\frac1{1+\|f_i-p_i\|_\infty}p_i(x_i),\dots}.
\end{align*}
The remainder of the proof is analogous to that of
Theorem~\ref{thm:approx-upper-F-f}, with the obvious notational
changes and an optimal choice of approximants $P,p_1,\dots,p_k.$
\end{proof}

\paragraph{Bounds using block sensitivity.}
Several results above can be sharpened somewhat using
the notion of \emph{block sensitivity}, denoted $\bs(F)$ for a
function $F\colon\mook\to\moo$ and defined as the maximum number
of nonempty disjoint subsets $S_1,S_2,S_3,\dots\subseteq\onetok$
such that on some input $x\in\mook,$ flipping the bits in any one
set $S_i$ changes the value of the function. We have:

\begin{proposition}
\label{prop:certificate-bs}
Let $F\colon\mook\to\moo$ be a given Boolean function. Let
$y\in\mook$ be a random string whose $i$th bit is
set to $-1$ with probability at most $\alpha,$
independently for each $i.$ 
Then for every $x\in\mook,$
\begin{align*}
\Prob_{y}[F(x_1,\dots,x_k)\ne F(x_1y_1,\dots,x_ky_k)] \leq
2\alpha\bs(F).
\end{align*}
\end{proposition}

\begin{proof}
By monotonicity, we may assume that each bit of $y$ takes on
$-1$ with probability exactly $\alpha.$
For a fixed integer $r$ and a uniformly random string $y\in\mook$ with
$\abs{\{i:y_i=-1\}}=r,$ the probability that $F(\dots,x_i,\dots)\ne
F(\dots,x_iy_i,\dots)$ is clearly at most $\bs(F)/\lfloor
k/r\rfloor\leq 2r\bs(F)/k.$ Averaging over $r$ gives the sought
bound.
\end{proof}

Since by definition $C(F)\geq\bs(F)$ for every function
$F\colon\mook\to\moo,$ use of
Proposition~\ref{prop:certificate-bs} instead of
Proposition~\ref{prop:certificate} can lead to sharper bounds in
some results of this section.  Specifically,
Theorems~\ref{thm:approx-dp}, \ref{thm:approx-upper-F-f},
\ref{thm:approx-dp-generalized}, and
\ref{thm:approx-upper-F-f-generalized} remain valid 
with (\ref{eqn:approx-dp}) replaced by
\begin{align}
\deg_{\epsilon+\eta-4\delta\bs(F)}(F(f,\dots,f)) &\geq \deg_{\epsilon}(F)
\deg_{1-\delta}(f);
\end{align}
with (\ref{eqn:approx-upper-F-f-nu}) and
(\ref{eqn:approx-upper-F-f-generalized}) replaced by 
\begin{align}
\eta(\Delta,\delta) = 
\Delta + \frac{4\delta\bs(F)}{1+\delta};
\end{align}
and with (\ref{eqn:approx-dp-generalized}) replaced by 
\begin{align}
\deg_{\epsilon+\eta-4\delta\bs(F)}(F(f_1,\dots,f_k)) 
&\geq \deg_{\epsilon,v}(F).
\end{align}
In particular, we obtain from Theorem~\ref{thm:approx-dp} that
\begin{align*}
\deg_{1/3}(F(f,\dots,f))&\geq \deg_{2/3}(F)\deg_{1-(12\bs(F))^{-1}}(f)
\nonumber\\
&\geq
\deg_{1/3}(F)\deg_{1/3}(f)\cdot \Omega\left(\frac1{1+\bs(F)}\right).
\end{align*}

\subsection{Auxiliary results on rational approximation
\label{sec:auxiliary}}

In this section, we prove a number of auxiliary facts about
uniform approximation and sign-representation. This preparatory
work will set the stage for our analysis of conjunctions of
functions.  We start by spelling out the exact relationship
between the rational approximation and sign-representation of a
Boolean function.

\begin{theorem}
\label{thm:trivial-approx}
Let $f\colon X\to\moo$ be a given function, where $X\subset\Re^n$ 
is finite. Then for every integer $d,$
\[ \degthr(f)\leq d \quad \Leftrightarrow \quad R^+(f,d)<1. \]
\end{theorem}

\begin{proof}
For the forward implication, let $p$ be a polynomial of degree at most $d$ 
such that $f(x)p(x)>0$ for every $x\in X.$ Letting 
$M=\max_{x\in X} |p(x)|$ and $m=\min_{x\in X}|p(x)|,$ we have 
\begin{align*}
R^+(f,d)\leq \max_{x\in
X}\left|f(x)-\frac{p(x)}{M}\right|\leq 1-\frac mM <1.
\end{align*}

For the converse, fix a degree-$d$ rational function $p(x)/q(x)$
with $q(x)>0$ on $X$ and $\max_X\abs{f(x) - \{p(x)/q(x)\}}
< 1.$ Then clearly $f(x)p(x)>0$ on $X.$
\end{proof}

Our next observation amounts to reformulating the rational approximation
of Boolean functions in a way that is more analytically pleasing.

\begin{theorem}
\label{thm:balance}
Let $f\colon X\to\moo$ be a given function, where $X\subset\Re^n$ 
is finite.  Then for every integer $d\geq\degthr(f),$ one has
\begin{align*}
R^+(f,d) = \inf_{c\geq1} \; \frac{c^2-1}{c^2+1},
\end{align*}
where the infimum is over all $c\geq1$ for which there exist
polynomials $p,q$ of degree up to $d$ such that
$0<\frac1c q(x) \leq f(x)p(x) \leq cq(x)$ on $X.$ 
\end{theorem}

\begin{proof}
In view of Theorem~\ref{thm:trivial-approx}, the quantity $R^+(f,d)$
is the infimum over all $\epsilon<1$ for which there exist polynomials
$p$ and $q$ of degree up to $d$ such that $0<(1-\epsilon)q(x)\leq
f(x)p(x)\leq (1+\epsilon)q(x)$ on $X.$ Equivalently, one may require
that
\begin{align*}
0<\frac{1-\epsilon}{\sqrt{1-\epsilon^2}}\, q(x)
\leq f(x)p(x)\leq 
\frac{1+\epsilon}{\sqrt{1-\epsilon^2}}\, q(x).
\end{align*}
Letting $c=c(\epsilon)=\sqrt{(1+\epsilon)/(1-\epsilon)},$ the
theorem follows.
\end{proof}



We will now show that if a degree-$d$ rational approximant
achieves error $\epsilon$ in approximating a given Boolean
function, then a degree-$2d$ approximant can achieve error as
small as $\epsilon^2.$ Note that this result is a refinement of
Theorem~\ref{thm:error-boosting} for small $k.$

\begin{theorem}
\label{thm:accuracy-boost}
Let $f\colon X\to\moo$ be a given function, where $X\subseteq\Re^n.$ 
Let $d$ be a given integer. Then
\[ R^+(f,2d) \leq \PARENS{\frac{\epsilon}{1 + \sqrt{1-\epsilon^2}}}^2,\]
where $\epsilon = R(f,d).$
\end{theorem}

\begin{proof}
The theorem is clearly true for $\epsilon=1.$ For $0\leq
\epsilon<1,$ consider the univariate rational function
\[ S(t) = \frac{4\sqrt{1-\epsilon^2}}{1+\sqrt{1-\epsilon^2}}\cdot 
          \frac{t}{t^2+(1-\epsilon^2)}. \]
Calculus shows that
\[ \max_{\rule{0mm}{8pt}1-\epsilon\leq |t|\leq 1+\epsilon}
|\sign t - S(t)| =
   \PARENS{\frac{\epsilon}{1 + \sqrt{1-\epsilon^2}}}^2.\]
Fix a sequence $A_1,A_2,\dots$ of rational functions of degree at
most $d$ such that $\sup_{x\in X}|f(x) - A_m(x)|\to\epsilon$ as
$m\to\infty.$ Then $S(A_1(x)),S(A_2(x)),\dots$ is the sought
sequence of approximants to $f,$ each a rational function of
degree at most $2d$ with a positive denominator.
\end{proof}

\begin{corollary}
\label{cor:general-amplify}
Let $f\colon X\to\moo$ be a given function, where $X\subseteq\Re^n.$ 
Then for all integers $d\geq1$ and reals $t\geq2,$ 
\[ R^+(f,td)\leq R(f,d)^{t/2}.  \]
\end{corollary}

\begin{proof}
If $t=2^k$ for some integer $k\geq1,$ then repeated applications of
Theorem~\ref{thm:accuracy-boost} yield
$R^+(f,2^kd) \leq R(f,2^{k-1}d)^2 \leq \cdots \leq R(f,d)^{2^k}.$
The general case follows because $2^{\lfloor\log t\rfloor}\geq t/2.$
\end{proof}

\subsection{Conjunctions of functions
\label{sec:intersection-two}}

In this section, we prove our direct product theorems for 
conjunctions of Boolean
functions.  Recall that a key challenge 
will be, given a sign-representation $\phi(x,y)$ of a
composite function $f(x)\wedge \g(y),$ to suitably break down $\phi$
and recover individual rational approximants of $f$ and $\g.$ We
now present an ingredient of our solution, namely, a certain fact
about pairs of matrices based on Farkas' Lemma.  For the time being,
we will formulate this fact in a clean and abstract way.

\begin{theorem}
\label{thm:zero-sum-game}
Fix matrices $A,B\in\Re^{m\times n}$ and a real $c\geq1.$ Consider
the following system of linear inequalities in $u,v\in\Re^n${\rm :}
\begin{equation}
\LEFTRIGHT.\rcbrace{
\hspace{3cm}
	\begin{aligned}
	\frac{1}{c}\, Au \leq &Bv \leq cAu,\\
	 u&\geq 0, \\
	 v&\geq 0. \\
	\end{aligned}
\hspace{3cm}}
  \label{eqn:matrix-system}
\end{equation}
If $u=v=0$ is the only solution to {\rm (\ref{eqn:matrix-system}),} then
there exist vectors $w\geq0$ and $z\geq0$ such that
\[ w\tr A + z\tr B > c (z\tr A + w\tr B). \]
\end{theorem}

\begin{proof}
If $u=v=0$ is the only solution to (\ref{eqn:matrix-system}), then
linear programming duality implies the existence of vectors $w\geq0$
and $z\geq 0$ such that $w\tr A > cz\tr A$ and $z\tr B > cw\tr B.$
Adding the last two inequalities completes the proof.
%
\end{proof}

For clarity of exposition, we first prove the main result of this
section for the case of \emph{two} Boolean functions at least one of which
is \emph{odd}.  While this case seems restricted, we will see that
it captures the full complexity of the problem.


\begin{theorem}
\label{thm:main-finite-odd}
Let $f\colon X\to\moo$ and $\g\colon Y\to\moo$ be given functions, where
$X,Y\subset\Re^n$ are arbitrary finite sets. Assume that $f\not\equiv1$
and $\g\not\equiv1.$ Let $d=\degthr(f\wedge \g).$ If $f$ is odd, then
\[ R^+(f,2d) + R^+(\g,d) < 1. \]
\end{theorem}

\begin{proof}
We first collect some basic observations. Since $f\not\equiv1$
and $\g\not\equiv1,$ we have $\degthr(f)\leq d$ and $\degthr(\g)\leq d.$ 
Therefore, Theorem~\ref{thm:trivial-approx} implies that
\begin{equation}
 R^+(f,d)<1,\qquad R^+(\g,d)<1. 
 \label{eqn:both-nontrivial}
\end{equation}
In particular, the theorem holds if $R^+(\g,d)=0.$ In the
remainder of the proof, we assume that $R^+(\g,d)=\epsilon,$ where
$0<\epsilon<1.$

By hypothesis, there exists a degree-$d$ polynomial $\phi$ 
such that $f(x)\wedge \g(y) = \sign \phi(x,y)$ for all $x\in X,$ $y\in Y.$
Define 
\[ X^- = \{ x\in X : f(x)=-1 \}. \]
Since $X$ is closed under negation and $f$ is odd, we have $f(x)=1
\Leftrightarrow -x\in X^-.$ We will make several uses of this fact
in what follows, without further mention.

Put 
\[ c = \SQRT{\frac{1 + (1-\delta)\epsilon}{1-(1-\delta)\epsilon}}, \]
where $\delta\in(0,1)$ is sufficiently small.
Since $R^+(\g,d) > (c^2-1)/(c^2+1),$ we know by Theorem~\ref{thm:balance}
that there cannot exist polynomials $p,q$ of degree up to $d$ such that
\begin{equation}
0 < \frac1c q(y) \leq \g(y)p(y) \leq cq(y), \qquad y\in Y.
\label{eqn:qpq}
\end{equation}
We claim, then, that there cannot exist reals $a_x\geq0,$ $x\in X,$ not all
zero, such that
\[ \frac1c \sum_{x\in X^-} a_{-x}\phi(-x,y) 
\leq \g(y) \sum_{x\in X^-} a_x \phi(x,y) 
\leq c\sum_{x\in X^-} a_{-x} \phi(-x,y), \quad y\in Y. \]
Indeed, if such reals $a_x$ were to exist, then (\ref{eqn:qpq}) would
hold for the polynomials $p(y) = \sum_{x\in X^-}
a_x\phi(x,y)$ and $q(y)=\sum_{x\in X^-} a_{-x}\phi(-x,y).$
In view of the nonexistence of the $a_x,$ Theorem~\ref{thm:zero-sum-game} 
applies to the matrices 
\[ \Big[\phi(-x,y)\Big]_{y\in Y,\, x\in X^-},\qquad
   \Big[\g(y)\phi(x,y)\Big]_{y\in Y,\, x\in X^-}        \]
and guarantees the existence of nonnegative reals $\lambda_y,\mu_y$ for
$y\in Y$ such that
\begin{multline}
 \sum_{y\in Y} \lambda_y \phi(-x,y) + 
   \sum_{y\in Y} \mu_y  \g(y) \phi(x,y)\\
   > c\PARENS{ 
   \sum_{y\in Y}^{~} \mu_y \phi(-x,y) + 
   \sum_{y\in Y} \lambda_y \g(y) \phi(x,y) }, \qquad x\in X^-.
\qquad
   \label{eqn:lambda-mu-claim}
\end{multline}
Define polynomials $\alpha,\beta$ on $X$ by
\begin{align*}
 \alpha(x) &= 
\sum_{y\in \g^{-1}(-1)} \{\lambda_y \phi(-x,y) - \mu_y\phi(x,y)\}, \\
\beta(x) &= 
\sum_{y\in \g^{-1}(1)\phantom{-}} \{\lambda_y \phi(-x,y) + \mu_y\phi(x,y)\}.
\end{align*}
Then (\ref{eqn:lambda-mu-claim}) can be restated as
\[ \alpha(x) + \beta(x) > c\{-\alpha(-x) + \beta(-x)\}, \qquad x\in X^-. \]
Both members of this inequality are nonnegative, and thus 
$\{\alpha(x) + \beta(x)\}^2 > c^2\{-\alpha(-x) + \beta(-x)\}^2$ for 
$x\in X^-.$
Since in addition $\alpha(-x)\leq0$ and $\beta(-x)\geq0$ for $x\in X^-,$ we have
\[ \{\alpha(x) + \beta(x)\}^2 > c^2\{\alpha(-x) + \beta(-x)\}^2, 
\qquad x\in X^-. \]
Letting $\gamma(x) = \{\alpha(x) + \beta(x)\}^2,$ we see that 
\[ R^+(f,2d) \leq \max_{x\in X} 
\left| f(x) - \frac{c^2+1}{c^2} \cdot
\frac{\gamma(-x) - \gamma(x)}{\gamma(-x)+\gamma(x)}\right|
\leq \frac1{c^2}<1-\epsilon, \]
where the final inequality holds for all $\delta\in(0,1)$ small enough.
\end{proof}

\begin{remark*}
In Theorem~\ref{thm:main-finite-odd} and elsewhere in this paper,
the degree of a multivariate polynomial $p(x_1,x_2,\dots,x_n)$ is
defined as the greatest total degree of any monomial of $p.$ A
related notion is the \emph{partial degree} of $p,$ which is the
maximum degree of $p$ in any one of the variables $x_1,x_2,\dots,x_n.$
One readily sees that the proof of Theorem~\ref{thm:main-finite-odd}
applies unchanged to this alternate notion. Specifically, if the
conjunction $f(x)\wedge \g(y)$ can be sign-represented by a polynomial
of partial degree $d,$ then there exist rational functions $F(x)$
and $G(y)$ of partial degree $2d$ such that $\|f-F\|_\infty +
\|\g-G\|_\infty<1.$ In the same way, the program of Section~\ref{sec:h}
carries over, with cosmetic changes, to the notion of partial degree.
Analogously, our proofs apply to hybrid definitions of degree, such
as partial degree over blocks of variables. Other, more abstract
notions of degree can also be handled.  In the remainder of the
paper, we will maintain our focus on total degree and will not
elaborate further on its generalizations.
\end{remark*}

As promised, we will now remove the assumption, made in
Theorem~\ref{thm:main-finite-odd}, about one of the functions being
odd.  The result that we are about to prove settles
Theorem~\ref{thm:main-two} from the Introduction.

\begin{theorem}
\label{thm:main-finite}
Let $f\colon X\to\moo$ and $\g\colon Y\to\moo$ be given functions, where
$X,Y\subset\Re^n$ are arbitrary finite sets. Assume that $f\not\equiv1$
and $\g\not\equiv1.$ Let $d=\degthr(f\wedge \g).$ Then 
\begin{align}
R^+(f,4d) + R^+(\g,2d) < 1 
\label{eqn:42}
\end{align}
and, by symmetry,
\begin{align*}
R^+(f,2d) + R^+(\g,4d) < 1.
\end{align*}
\end{theorem}

\begin{proof}
It suffices to prove (\ref{eqn:42}).
Define $X'\subset\Re^{n+1}$ by $X'= \{(x,1),(-x,-1) : x\in X\}.$
It is clear that $X'$ is closed under negation. Let $f'\colon X'\to\moo$ be
the odd Boolean function given by 
\[ f'(x,b) = \ccases{
f(x), &b=1,\\
-f(-x), &b=-1.
}
\]
Let $\phi$ be a polynomial of degree no greater than $d$
such that $f(x)\wedge \g(y)\equiv
\sign \phi(x,y).$ Fix an input $\tilde x\in X$ such that 
$f(\tilde x)=-1.$ Then 
 $f'(x,b) \wedge \g(y) \equiv \sign\BRACES{ 
   K(1+b)\phi(x,y) + \phi(-x,y)\phi(\tilde x,y)
} 
$ for a large enough constant $K\gg1,$ whence 
\[ \degthr(f'\wedge \g) \leq 2d. \]
Theorem~\ref{thm:main-finite-odd} now yields 
$R^+(f',4d) + R^+(\g,2d) < 1.$
Since $R^+(f,4d) \leq R^+(f',4d)$ by definition, the proof is complete.
\end{proof}

Finally, we obtain an analogue of Theorem~\ref{thm:main-finite} for
a conjunction of three and more functions.

\begin{theorem}
\label{thm:main-finite-multi}
Let $f_1,f_2,\dots,f_k$ be given Boolean functions on finite sets
$X_1,X_2,\dots,X_k$ $\subset \Re^n,$ respectively. 
Assume that $f_i\not\equiv1$ for $i=1,2,\dots,k.$
Let $d=\degthr(f_1\wedge f_2\wedge \cdots \wedge f_k).$ Then 
\[ \sum_{i=1}^k R^+(f_i,D) < 1 \]
for $D=8d\log2k.$
\end{theorem}

\begin{proof}
Since $f_1,f_2,\dots,f_k\not\equiv 1,$ it follows that 
for each pair of indices $i<j,$ the function $f_i\wedge f_j$ is a
subfunction of $f_1\wedge f_2\wedge \cdots \wedge f_k.$
Theorem~\ref{thm:main-finite} now shows that for each $i<j,$
\begin{align}
R^+(f_i,4d)  + R^+(f_j,4d) < 1. \label{eqn:ij}
\end{align}
Without loss of generality, 
$R^+(f_1,4d)=\max_{i=1,\dots,k} R^+(f_i,4d).$ 
Abbreviate $\epsilon=R^+(f_1,4d).$
By (\ref{eqn:ij}), 
\[R^+(f_i,4d) < \min\BRACES{1-\epsilon,\frac12},
  \qquad i=2,3,\dots,k.\]
Now Corollary~\ref{cor:general-amplify} implies that 
\begin{align*}
\sum_{i=1}^k R^+(f_i,D)
 \leq \epsilon + \sum_{i=2}^k R^+(f_i,4d)^{1 + \log k}
 < 1. 
 \tag*{\qedhere}
\end{align*}
\end{proof}

\subsection{Other combining functions \label{sec:h}}

As we will now see, the development in Section~\ref{sec:intersection-two}
applies to many combining functions other than conjunctions.
Disjunctions are an illustrative starting point.  Consider two
Boolean functions $f\colon X\to\moo$ and $\g\colon Y\to\moo,$ where
$X,Y\subset\Re^n$ are finite sets and $f,\g\not\equiv-1.$ Let 
$d=\degthr(f\vee \g).$ Then, we claim that
\begin{equation}
 R^+(f,4d) + R^+(\g,4d) < 1.
 \label{eqn:f-or-g}
\end{equation}
To see this, note first that the function $f\vee \g$ has the same
threshold degree as its negation, $\overline f\wedge \overline \g.$
Applying Theorem~\ref{thm:main-finite} to the latter function shows
that
\[ 
   R^+(\overline f,4d) + R^+(\overline \g,4d) < 1. 
\]
This is equivalent to (\ref{eqn:f-or-g}) since 
approximating a function is the same as approximating its negation:
$R^+(\overline f,4d)=R^+(f,4d)$ and 
$R^+(\overline \g,4d)=R^+(\g,4d).$ As in the case
of conjunctions, (\ref{eqn:f-or-g}) can be strengthened to
\[ 
   R^+(f,2d) + R^+(\g,2d) < 1
\]
if at least one of $f,\g$ is known to be odd. These observations
carry over to disjunctions of multiple functions, 
$f_1\vee f_2\vee \cdots\vee f_k.$

The above discussion is still too specialized. In what follows, we
consider composite functions $h(f_1,f_2,\dots,f_k),$
where $h\colon\mook\to\moo$ is any given Boolean function.  We will
shortly see that the results of the previous sections hold for
various $h$ other than $h=\AND$ and $h=\OR.$

We start with some notation and definitions.  
Let $f,h\colon \mook\to\moo$
be given Boolean functions. Recall that $f$ is called a \emph{subfunction}
of $h$ if for some fixed strings $y,z\in\mook,$ one has
\[ f(x) = h(\dots,(x_i\wedge y_i) \vee z_i,\dots) \]
for each $x\in\mook.$ In words, $f$ can
be obtained from $h$ by replacing some of the variables
$x_1,x_2,\dots,x_k$ with fixed values ($-1$ or $+1$).

\begin{definition}
\label{def:and-reducible}
A function $F\colon \mook\to\moo$ is \emph{\AND-reducible} if for each
pair of indices $i,j,$ where $1\leq i\leq j\leq k,$ at least 
one of the eight functions
\[ 
\begin{aligned}
x_i&\wedge x_j,\\
   x_i&\wedge \overline {x_j},\\
   \overline {x_i}&\wedge x_j,\\
   \overline {x_i}&\wedge \overline {x_j},
\end{aligned}
\quad\qquad
\begin{aligned}
  x_i&\vee x_j,\\
   x_i&\vee \overline {x_j},\\
   \overline {x_i}&\vee x_j,\\
   \overline {x_i}&\vee \overline {x_j}
\end{aligned}
\]
is a subfunction of $F(x).$ 
\end{definition}

\begin{theorem}
\label{thm:and-reducible}
Let $f_1,f_2,\dots,f_k$ be nonconstant Boolean functions on finite
sets $X_1,X_2,\dots,X_k\subset\Re^n,$ respectively.  Let $F\colon \mook\to\moo$
be an \AND-reducible function.  Put
$d=\degthr(F(f_1,f_2,\dots,f_k)).$ Then
\[ \sum_{i=1}^k R^+(f_i,D) < 1 \]
for $D=8d\log 2k.$
\end{theorem}

\begin{proof}
Since $F$ is \AND-reducible, it follows that for each pair of indices
$i<j,$ one of the following eight functions is a subfunction of
$F(f_1,\dots,f_k)$:
\[ 
\begin{aligned}
f_i&\wedge f_j,\\
   f_i&\wedge \overline {f_j},\\
   \overline {f_i}&\wedge f_j,\\
   \overline {f_i}&\wedge \overline {f_j},
\end{aligned}
\quad\qquad
\begin{aligned}
  f_i&\vee f_j,\\
   f_i&\vee \overline {f_j},\\
   \overline {f_i}&\vee f_j,\\
   \overline {f_i}&\vee \overline {f_j}.
\end{aligned}
\]
By Theorem~\ref{thm:main-finite} (and the opening remarks of this
section), 
\[ R^+(f_i,4d) + R^+(f_j,4d) < 1. \]
The remainder of the proof is identical to the proof of
Theorem~\ref{thm:main-finite-multi}, starting at equation (\ref{eqn:ij}).
\end{proof}

In summary, the development in Section~\ref{sec:intersection-two}
naturally extends to compositions
$F(f_1,f_2,\dots,f_k)$ for various $F.$ For a function
$F\colon \mook\to\moo$ to be \AND-reducible, $F$ must clearly depend
on all of its inputs. This necessary condition
is often sufficient, for example when $F$ is a read-once
AND/OR/NOT formula or a halfspace.  Hence, Theorem~\ref{thm:main-h}
from the Introduction is a corollary of Theorem~\ref{thm:and-reducible}.

\begin{remark*}
If more information is available about the combining function $F,$
Theorem~\ref{thm:and-reducible} can be generalized to let some of
$f_1,\dots,f_k$ be constant functions. For example, some or all of
the functions $f_1,\dots,f_k$ in Theorem~\ref{thm:main-finite-multi} can
be identically true. Another direction for generalization is as
follows. In Definition~\ref{def:and-reducible}, one considers all
the ${k\choose 2}$ distinct pairs of indices $(i,j).$ If one happens
to know that $f_1$ is harder to approximate than $f_2,\dots,f_k,$
then one can relax Definition~\ref{def:and-reducible} to examine
only the $k-1$ pairs $(1,2),(1,3),\dots,(1,k).$ We do not formulate
these extensions as theorems, the fundamental technique being already
clear.
\end{remark*}

\subsection{Additional observations}
\label{sec:additional-conj}

Analogous to Section~\ref{sec:dp},
our results here can be viewed as a technique for proving lower
bounds on the threshold degree of composite functions 
$F(f_1,f_2,\dots,f_k).$ We make this view explicit in the following
statement, which is the contrapositive of Theorem~\ref{thm:and-reducible}.

\begin{theorem}
\label{thm:and-reducible-degthr-technique}
Let $f_1,f_2,\dots,f_k$ be nonconstant Boolean functions on finite
sets $X_1,X_2,\dots,X_k\subset\Re^n,$ respectively.  Let $F\colon \mook\to\moo$
be an $\AND$-reducible function.  Suppose that
$\sum R^+( f_i,D) \geq 1$
for some integer $D.$ Then
\begin{equation}
\degthr(F(f_1,f_2,\dots,f_k))>\frac D{8\log 2k}.
\label{eqn:lower-bound-degthr}
\end{equation}
\end{theorem}

\begin{remark}[On the tightness of
Theorem~\textup{\ref{thm:and-reducible-degthr-technique}}]
\label{rem:tightness}
Theorem~\ref{thm:and-reducible-degthr-technique} is close to
optimal. For example, when $F=\AND,$ the lower bound
in~(\ref{eqn:lower-bound-degthr}) is tight up to a factor of
$\Theta(k\log k).$ This can be seen by the well-known
argument~\cite{beigel91rational} described in the Introduction.
Specifically, fix an integer $D$ such that $\sum R^+(f_i,D) <1.$
Then there exists a rational function $p_i(x_i)/q_i(x_i)$ on $X_i,$
for $i=1,2,\dots,k,$ such that $q_i$ is positive on $X_i$ and
\[ \sum_{i=1}^k \;\max_{x_i\in X_i} \left| f_i(x_i) -
\frac{p_i(x_i)}{q_i(x_i)} \right| < 1. \]
As a result,
\begin{align*}
\bigwedge_{i=1}^k f_i(x_i) 
	 \equiv \sign\PARENS{k-1 + \sum_{i=1}^k f_i(x_i) }
	\equiv \sign\PARENS{ k-1 + \sum_{i=1}^k 
	\frac{p_i(x_i)}{q_i(x_i)}}. 
\end{align*}
Multiplying by $\prod q_i(x_i)$ yields
\begin{align*}
 \bigwedge_{i=1}^k f_i(x_i) \equiv \sign\PARENS{ 
(k-1)\prod_{i=1}^k q_i(x_i) +
\sum_{i=1}^k p_i(x_i)\prod_{j\in\{1,\dots,k\}\setminus\{i\}}{q_j(x_j)}
},
\end{align*}
whence $\degthr(f_1\wedge f_2\wedge \cdots \wedge f_k) \leq k D.$
This settles our claim regarding $F=\AND.$ For arbitrary
$\AND$-reducible functions $F\colon \mook\to\moo,$ a similar
argument (cf.~Theorem~31 of Klivans et al.~\cite{KOS:02}) shows
that the lower bound in~(\ref{eqn:lower-bound-degthr}) is tight up
to a polynomial in $k.$
\end{remark}

\bigskip
We close this section with one additional result.

\begin{theorem}
\label{thm:2-to-k}
Let $f\colon X\to\moo$ be a given function, where $X\subset\Re^n$ is
finite. Then for every integer $k\geq2,$
\begin{equation}
\degthr(\undercbrace{f\wedge f\wedge \cdots\wedge f}_k)
\leq (8k\log k)\cdot \degthr(f\wedge f).
\label{eqn:2-to-k}
\end{equation}
\end{theorem}

\begin{proof}
Put $d=\degthr(f\wedge f).$ Theorem~\ref{thm:main-finite} implies
that $R^+(f,4d)< 1/2,$ whence $R^+(f,8d\log k)<1/k$ by
Corollary~\ref{cor:general-amplify}.  By 
the argument in Remark~\ref{rem:tightness}, this proves the theorem.
\end{proof}

To illustrate, let $\Ccal$ be a given class of functions on
$\moon,$ such as halfspaces. Theorem~\ref{thm:2-to-k} shows that
the task of constructing a sign-representation for the
intersections of up to $k$ members from $\Ccal$ reduces to the
case $k=2.$ In other words, solving the problem for $k=2$
essentially solves it for all $k.$ The dependence on $k$ in
(\ref{eqn:2-to-k}) is tight up to a factor of $16\log k,$ even in
the simple case when $f$ is the OR
function~\cite{minsky88perceptrons}.

\section{Rational approximation of a halfspace}

In this section, we determine how well a rational function of any
given degree can approximate the canonical halfspace.
The lower bounds in Theorem~\ref{thm:main-approx-hs},
the main result to be proved in this section, 
are considerably more involved than the upper bounds. 
To help build some intuition in the former case, we first obtain
the upper bounds (Section~\ref{sec:small-error}) and only then
prove the lower bounds (Sections~\ref{sec:preparatory}
and~\ref{sec:rational}).

\subsection{Upper bounds}
\label{sec:small-error}

As shown in the Introduction, the OR function on $n$ bits has
$R^+(\OR,1)=0.$ A similar example is the ODD-MAX-BIT 
function $f\colon\zoon\to\moo,$ due to Beigel~\cite{beigel94perceptrons},
defined by
\begin{align*}
f(x) = \sign\left(1 + \sum_{i=1}^n (-2)^i x_i\right).
\end{align*}
Indeed, letting
\begin{align*}
A_M(x) = \frac{1 + \sum_{i=1}^n (-M)^i x_i}
              {1 + \sum_{i=1}^n   M^i  x_i},
\end{align*}
we have $\|f - A_M\|_\infty\to0$ as $M\to\infty.$ Thus,
$R^+(f,1) = 0.$ With this construction in mind, we now turn to
the canonical halfspace. We start with an auxiliary result that
generalizes the argument just given.

\begin{lemma}
  \label{lem:012-dfa}
Let $f\colon\{0,\pm1,\pm2\}^n\to\moo$ be the function
given by $f(z) = \sign(1+\sum_{i=1}^n
2^iz_i).$ Then
\begin{align*} 
R^+(f,64) = 0.
\end{align*}
\end{lemma}

\begin{proof}
Consider the deterministic finite automaton in
Figure~\ref{fig:dfa}.
The automaton has two terminal states (labeled ``$+$'' and ``$-$'')
and three nonterminal states (the start state and two additional
states). We interpret the output of the automaton to be $+1$ and
$-1$ at the two terminal states, respectively, and $0$ otherwise.
A string $z=(z_n,z_{n-1},\dots,z_1,0)\in\{0,\pm1,\pm2\}^{n+1},$
when read by the automaton left to right, forces it to output
exactly $\sign(\sum_{i=1}^n 2^i z_i).$ If the automaton is
currently at a nonterminal state, this state is determined
uniquely by the last two symbols read.  Hence, the output of the
automaton on input 
$z=(z_n,z_{n-1},\dots,z_1,0)\in\{0,\pm1,\pm2\}^{n+1}$ is given by
\begin{align*}
\sign\PARENS{\sum_{i=0}^n 2^i \alpha(z_{i+2},z_{i+1},z_i)}
\end{align*}
for a suitable map $\alpha\colon\{0,\pm1,\pm2\}^3\to\{0,-1,+1\},$
where we adopt the shorthand $z_{n+1}=z_{n+2}=z_0=0.$ 
Put 
\begin{align*}
A_M(z) = \frac{ 1 + \sum_{i=0}^n M^{i+1} \alpha(z_{i+2},z_{i+2},z_i) }
           { 1 + \sum_{i=0}^n M^{i+1} \abs{\alpha(z_{i+2},z_{i+2},z_i)} }.
\end{align*}
By interpolation, the numerator and denominator of $A_M$ can be represented
by polynomials of degree no more than $4\times 4\times 4=64.$ On the
other hand, we have $\|f-A_M\|_\infty\to 0$ as $M\to\infty.$
\end{proof}

\begin{figure}[t]
\centering
\includegraphics[width=0.8\textwidth]{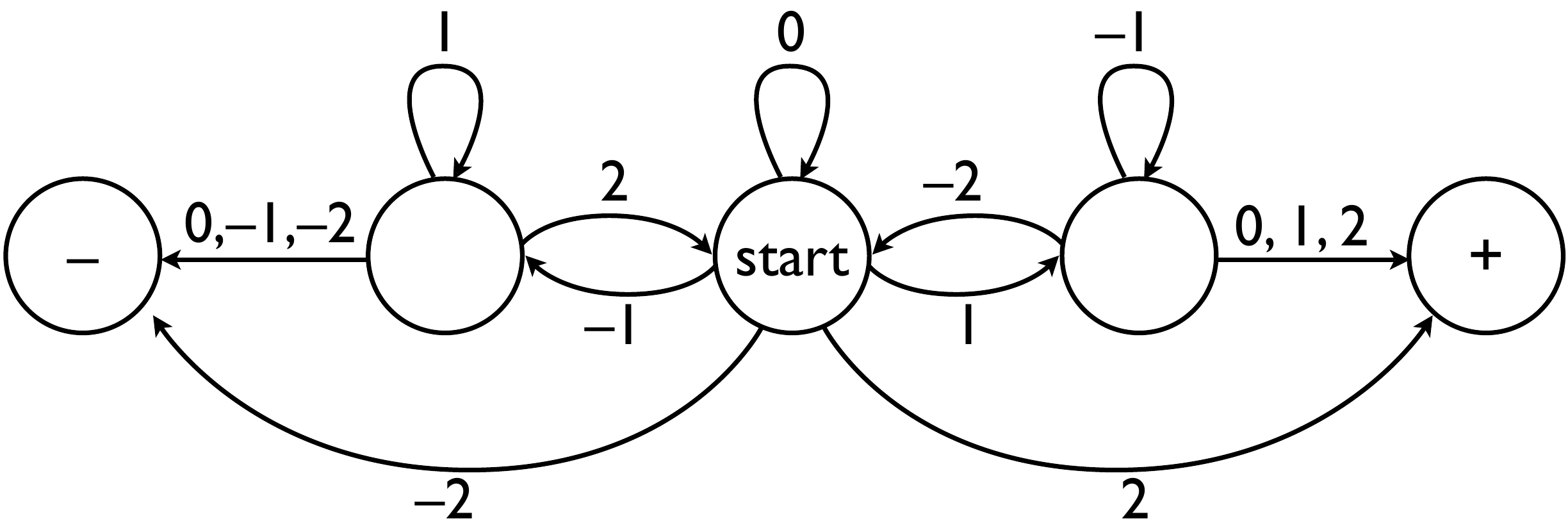}
\caption{Finite automaton for the proof of Lemma~\ref{lem:012-dfa}.}
\label{fig:dfa}
\end{figure}

We are now prepared to prove our desired upper bounds for halfspaces.

\begin{theorem}
\label{thm:upper-fnk}
Let $f\colon\moo^{nk}\to\moo$ be the function
given by 
\begin{align}
f(x) = \sign\left(1+\sum_{i=1}^n\sum_{j=1}^k 2^ix_{ij}\right).
\label{eqn:f-n-k}
\end{align}
Then 
\begin{align} 
R^+(f,64k\lceil \log k\rceil+1) =0.
\label{eqn:R0}
\end{align}
In addition, for all integers $d\geq1,$
\begin{align}
R^+(f,d) \leq 1 - (k2^{n+1})^{-1/d}.
\label{eqn:Rd}
\end{align}
\end{theorem}

\noindent
In particular, Theorem~\ref{thm:upper-fnk} settles all upper bounds on
$\rdeg_\epsilon(f)$ in Theorem~\ref{thm:main-approx-hs}.

\begin{proof}[Proof of Theorem \textup{\ref{thm:upper-fnk}}]
Theorem~\ref{thm:newman-approx} immediately implies
(\ref{eqn:Rd}) in view of the representation (\ref{eqn:f-n-k}).
It remains to prove (\ref{eqn:R0}). In the degenerate case $k=1,$ we have
$f\equiv x_{n1}$ and thus (\ref{eqn:R0}) holds. In what follows, we assume
that $k\geq2$ and put $\Delta=\lceil\log k\rceil.$
We adopt the convention that $x_{ij}\equiv 0$ for $i>n.$ For
$\ell=0,1,2,\dots,$ define
\begin{align*}
S_\ell 
 = \sum_{i=1}^{\Delta}\sum_{j=1}^k 2^{i-1} x_{\ell\Delta+i,j}.
\end{align*}
Then 
\begin{multline}
\sum_{i=1}^n \sum_{j=1}^k 2^{i-1}x_{ij} = 
\PARENS{
S_0 + 2^{2\Delta} S_2 + 2^{4\Delta} S_4 + 2^{6\Delta} S_6 +
\cdots}\\
+\PARENS{
2^{\Delta}S_1 + 2^{3\Delta} S_3 + 2^{5\Delta} S_5 + 2^{7\Delta} S_7 +
\cdots}.
\label{eqn:integer-sum}
\end{multline}
Now, each $S_\ell$ is an integer in $[-2^{2\Delta}+1,2^{2\Delta}-1]$
and therefore admits a representation as
\begin{align*}
S_\ell = z_{\ell,1} + 2z_{\ell,2} + 2^2z_{\ell,3} + \cdots +
2^{2\Delta-1}z_{\ell,2\Delta},
\end{align*}
where $z_{\ell,1},\dots,z_{\ell,2\Delta}\in\{-1,0,+1\}.$
Furthermore, each $S_\ell$ only depends on $k\Delta$ of the
original variables $x_{ij},$ whence $z_{\ell,1},\dots,z_{\ell,2\Delta}$
can all be viewed as polynomials of degree at most $k\Delta$ in the
original variables. Rewriting (\ref{eqn:integer-sum}),
\begin{align*}
\sum_{i=1}^n \sum_{j=1}^k 2^{i-1}x_{ij} 
&= \left(\sum_{i\geq 1} 2^{i-1} z_{\ell(i),j(i)}\right)
+ \left(\sum_{i\geq \Delta+1} 2^{i-1} z_{\ell'(i),j'(i)}\right)
\end{align*}
for appropriate indexing functions $\ell(i),\ell'(i),j(i),j'(i).$ Thus,
\begin{align*}
f(x) \equiv \sign\left(
1 + \sum_{i=1}^\Delta 2^i \undercbrace{z_{\ell(i),j(i)}}
 + \sum_{i\geq \Delta+1}2^i \undercbrace{\left(z_{\ell(i),j(i)} +
 z_{\ell'(i),j'(i)}\right)}
 \right).
\end{align*}
Since the underbraced expressions range in $\{0,\pm1,\pm2\}$ and are
polynomials of degree at most $k\Delta$ in the original variables,
Lemma~\ref{lem:012-dfa} implies (\ref{eqn:R0}).
\end{proof}

\subsection{Preparatory work} \label{sec:preparatory}

This section sets the stage for our rational approximation 
lower bounds with some preparatory
results about halfspaces. It will be convenient to establish some
additional notation, for use in this section only.  Here, we 
typeset real vectors in boldface ($\xbold_1,\xbold_2,\zbold,\vbold$) to better
distinguish them from scalars.  The $i$th component of a vector
$\xbold\in \Re^n$ is denoted by $(\xbold)_i,$ while the symbol $\xbold_i$
is reserved for another \emph{vector} from some enumeration.  In
keeping with this convention, we let $\ebold_i$ denote the vector with
$1$ in the $i$th component and zeroes everywhere else.  For
$\xbold,\ybold\in\Re^n,$ the vector $\xbold\ybold\in\Re^n$ is given by
$(\xbold\ybold)_i\equiv (\xbold)_i(\ybold)_i.$ More generally, for a polynomial
$p$ on $\Re^k$ and vectors $\xbold_1,\dots,\xbold_k\in\Re^n,$ we define
$p(\xbold_1,\dots,\xbold_k)\in\Re^n$ by $(p(\xbold_1,\dots,\xbold_k))_i =
p((\xbold_1)_i,\dots,(\xbold_k)_i).$ The expectation of a random variable
$\xbold\in\Re^n$ is defined componentwise, i.e., the vector
$\Exp[\xbold]\in\Re^n$ is given by $(\Exp[\xbold])_i\equiv \Exp[(\xbold)_i].$


For convenience, we adopt the notational shorthand $\alpha^0=1$
for all $\alpha\in\Re.$ 
In particular, if $\xbold\in \Re^n$ is a
given vector, then $\xbold^0 = (1,1,\dots,1)\in \Re^n.$
A scalar $\alpha\in\Re,$ when interpreted as a vector, stands for
$(\alpha,\alpha,\dots,\alpha).$ This shorthand allows one to speak of
$\Span\{1,\zbold,\zbold^2,\dots,\zbold^k\},$ for example, where $\zbold\in\Re^n$ is
a given vector.

\begin{theorem}
\label{thm:mu-b}
Let $N$ and $m$ be positive integers. Then reals
$\alpha_0,\alpha_1,\dots,\alpha_{4m}$ exist with the following
property: for each $\bbold\in\zoo^N,$ there is a probability
distribution $\mu_\bbold$ on $\{0,\pm1,\dots,\pm m\}^N$ such that
\begin{align*}
\Exp_{\vbold\sim \mu_\bbold}[(2\vbold + \bbold)^d] = 
	(\alpha_d,\alpha_d,\dots,\alpha_d), 
	\qquad d=0,1,2,\dots,4m.
\end{align*}
\end{theorem}

\begin{proof}
Let $\lambda_0$ and $\lambda_1$ be the distributions on
$\{0,\pm1,\dots,\pm m\}$ given by
\begin{align*}
  \lambda_0(t) = 16^{-m}{4m+1 \choose 2m+2t}, \qquad
  \lambda_1(t) = 16^{-m}{4m+1 \choose 2m+2t+1}.
\end{align*}
Then for $d=0,1,\dots,4m,$ one has
\begin{multline}
\Exp_{t\sim\lambda_0}[(2t)^d] - \Exp_{t\sim\lambda_1}[(2t+1)^d]\\
= 16^{-m} \sum_{t=0}^{4m+1} (-1)^t {4m+1\choose t}(t-2m)^d
=0, \qquad 
\label{eqn:indisting}
\end{multline}
where (\ref{eqn:indisting}) holds by Fact~\ref{fact:comb}.
Now, let 
$
\mu_\bbold=
\lambda_{(\bbold)_1} \times
\lambda_{(\bbold)_2} \times
\cdots
\times \lambda_{(\bbold)_N}.
$
Then in view of (\ref{eqn:indisting}), the theorem holds by letting
$\alpha_d = \Exp_{\lambda_0}[(2t)^d]$ for $d=0,1,2,\dots,4m.$
\end{proof}

Using the previous theorem, we will now establish another auxiliary
result pertaining to halfspaces.

\begin{theorem}
\label{thm:distribution}
Put $\zbold=(-2^n,-2^{n-1},\dots,-2^0,2^0,\dots,2^{n-1},2^n)\in \Re^{2n+2}.$ 
There are random
variables $\xbold_1,\xbold_2,\dots,\xbold_{n+1}\in
\{0,\pm1,\pm2,\dots,\pm (3n+1)\}^{2n+2}$ such that:
\begin{align}
\sum_{i=1}^{n+1} 2^{i-1}\xbold_i \equiv \zbold 
  \label{eqn:support}
\end{align}
and
\begin{align}
\Exp\left[\prod_{i=1}^n \xbold_i^{d_i}\right]
 \in\Span\{(1,1,\dots,1)\}
  \label{eqn:span-property}
\end{align}
for $d_1,\dots,d_n\in \{0,1,\dots,4n\}.$
\end{theorem}

\begin{proof}
Let 
\begin{align*}
 \xbold_i = 2\ybold_i-\ybold_{i-1} +\ebold_{n+1+i}-\ebold_{n+2-i}, \qquad
i=1,2,\dots,n+1,
\end{align*}
where $\ybold_0,\ybold_1,\dots,\ybold_{n+1}$ are suitable random variables
with $\ybold_0\equiv\ybold_{n+1}\equiv 0.$ Then property~(\ref{eqn:support})
is immediate. We will construct $\ybold_0,\ybold_1,\dots,\ybold_{n+1}$
such that the remaining property~(\ref{eqn:span-property}) holds as well.

Let $N=2n+2$ and $m=n$ in Theorem~\ref{thm:mu-b}. Then reals
$\alpha_0,\alpha_1,\dots,\alpha_{4n}$ exist with the property that for each
$\bbold\in\zoo^{2n+2},$ a probability distribution $\mu_\bbold$ can be
found on $\{0,\pm1,\dots,\pm n\}^{2n+2}$ such that
\begin{align}
\Exp_{\vbold\sim\mu_\bbold} [(2\vbold +\bbold)^d]=\alpha_d(1,1,\dots,1), 
  \qquad d=0,1,\dots,4n.
\label{eqn:alpha_d}
\end{align}
Now, we will specify the distribution of $\ybold_0,\ybold_1,\dots,\ybold_n$ by 
giving an algorithm for generating $\ybold_i$ from $\ybold_{i-1}.$ First, recall
that $\ybold_0\equiv \ybold_{n+1}\equiv 0.$ 
The algorithm for generating $\ybold_{i}$ given
$\ybold_{i-1}$ $(i=1,2,\dots,n)$ is as follows.
\begin{itemize}
\item [\textup{(1)}] Let $\ubold$ be the unique integer vector
such that $2\ubold - \ybold_{i-1} + \ebold_{n+1+i} - \ebold_{n+2-i} \in\zoo^{2n+2}.$ 
\item [\textup{(2)}] Draw a random vector $\vbold\sim \mu_{\bbold},$
where $\bbold=2\ubold - \ybold_{i-1} + \ebold_{n+1+i} - \ebold_{n+2-i}.$ 
\item [\textup{(3)}] Set $\ybold_i = \vbold + \ubold.$
\end{itemize}
One easily verifies that 
$\ybold_0,\ybold_1,\dots,\ybold_{n+1}\in\{0,\pm1,\dots,\pm 3n\}^{2n+2}.$

Let $R$ denote the resulting joint distribution of
$(\ybold_0,\ybold_1,\dots,\ybold_{n+1}).$ Let $i\leq n.$ Then conditioned
on any fixed value of $(\ybold_0,\ybold_1,\dots,\ybold_{i-1})$ in the support
of $R,$ the random variable $\xbold_i$ is by definition independent of
$\xbold_1,\dots,\xbold_{i-1}$ and is distributed identically to $2\vbold
+\bbold,$ for some fixed vector $\bbold\in\zoo^{2n+2}$ and a random
variable $\vbold\sim \mu_\bbold.$ In view of (\ref{eqn:alpha_d}), we
conclude that
\begin{align*}
\Exp\left[\prod_{i=1}^n \xbold_i^{d_i}\right]
 = (1,1,\dots,1) \prod_{i=1}^n \alpha_{d_i} 
\end{align*}
for all $d_1,d_2,\dots,d_n\in\{0,1,\dots,4n\},$
which establishes (\ref{eqn:span-property}).  It remains to note that
$\xbold_1,\xbold_2,\dots,\xbold_n\in\{-2n,-2n+1,\dots,-1,0,1,\dots,
2n,2n+1\}^{2n+2},$ whereas $\xbold_{n+1}= -\ybold_n
+\ebold_{2n+2}-\ebold_1\in\{0,\pm1,\dots,\pm(3n+1)\}^{2n+2}.$
\end{proof}

At last, we arrive at the main theorem of this section, which will play a
crucial role in our analysis of the rational approximation of halfspaces. 

\begin{theorem}
\label{thm:aux-thr-deg}
For $i=0,1,2,\dots,n,$ define
\begin{align*}
A_i &= \BRACES{(x_1,\dots,x_{n+1})\in\{0,\pm1,\dots,\pm(3n+1)\}^{n+1}
  \colon \quad
\sum_{j=1}^{n+1} 2^{j-1}x_j = 2^i}.
\end{align*}
Let $p(x_1,\dots,x_{n+1})$ be a real polynomial with sign $(-1)^i$
throughout $A_i$ $(i=0,1,2,\dots,n)$ and sign $(-1)^{i+1}$ throughout
$-A_i$ $(i=0,1,2,\dots,n).$ Then
\begin{align*}
 \deg p \geq 2n+1. 
\end{align*}
\end{theorem}

\begin{proof}
For the sake of contradiction, suppose that $p$ has degree no greater
than $2n.$ Put $\zbold=(-2^n,-2^{n-1},\dots,-2^0,2^0,\dots,2^{n-1},2^n).$
Let $\xbold_1,\dots,\xbold_{n+1}$ be the random variables constructed in
Theorem~\ref{thm:distribution}.  By (\ref{eqn:span-property}) and
the identity $\xbold_{n+1}\equiv2^{-n}\zbold - \sum_{i=1}^n 2^{i-n-1}\xbold_i,$
we have
\begin{align*}
\Exp[p(\xbold_1,\dots,\xbold_{n+1})] \in \Span\{1,\zbold,\zbold^2,\dots,\zbold^{2n}\}, 
\end{align*}
whence 
$
   \Exp[p(\xbold_1,\dots,\xbold_{n+1})] = q(\zbold) 
$
for a univariate polynomial $q\in P_{2n}.$ 
In view of (\ref{eqn:support}) and the
assumed sign behavior of $p,$ we have $\sign q(2^i) = (-1)^i$ and $\sign
q(-2^i) = (-1)^{i+1},$ for $i=0,1,2,\dots,n.$ Therefore, $q$ has
at least $2n+1$ roots. Since $q\in P_{2n},$ we arrive at a
contradiction. It follows that the assumed polynomial $p$ does
not exist.
\end{proof}

\begin{remark}
The passage $p\mapsto q$ in the proof of Theorem~\ref{thm:aux-thr-deg}
is precisely the linear degree-nonincreasing map
$M\colon\Re[x_1,x_2,\dots,x_{n+1}]\to\Re[x]$ described previously
in the Introduction.
\end{remark}

\subsection{Lower bounds}
\label{sec:rational}

The purpose of this section is to prove that the canonical halfspace 
cannot be approximated well by a rational function of low degree. A
starting point in our discussion is a criterion for inapproximability
by low-degree rational functions, which is applicable not only to
halfspaces but any odd Boolean functions on Euclidean space.

\newcommand{\X}{(x)}
\newcommand{\Xm}{(-x)}

\begin{theorem}[Criterion for inapproximability]
\label{thm:rational-degree-criterion-hs}
Fix a nonempty finite subset $S\subset\Re^m$ with $S\cap -S=\varnothing.$
Define $f\colon S\cup -S\to\moo$ by
\begin{align*}
f\X = \ccases{
+1, & x\in S, \\
-1, & x\in -S.
}
\end{align*}
Let $\psi$ be a real function such that
\begin{align}
\psi\X > 
\delta |\psi\Xm|, && x\in S,
\label{eqn:r-unbalanced-hs}
\end{align}
for some $\delta\in(0,1)$ and 
\begin{align}
\sum_{S\cup -S} \psi\X u\X = 0
\label{eqn:r-orthog}
\end{align}
for every polynomial $u$ of degree at most $d.$ Then
\begin{align*}
R^+(f,d) \geq \frac{2\delta}{1 + \delta}.
\end{align*}
\end{theorem}

\begin{proof}
Fix polynomials $p,q$ of degree at most $d$ such that $q$ is positive on
$S\cup -S.$ Put
\[ \epsilon = \max_{S\cup -S}
    \left| f\X - \frac{p\X}{q\X} \right|.  \]
We assume that $\epsilon<1$ since otherwise there is nothing to show.
For $x\in S,$ 
\begin{align}
(1-\epsilon)q\X \leq p\X \leq (1+\epsilon)q\X
\label{eqn:pqeps1-hs}
\end{align}
and
\begin{align}
(1-\epsilon)q\Xm \leq -p\Xm \leq (1+\epsilon)q\Xm.
\label{eqn:pqeps2-hs}
\end{align}
Consider the polynomial $u\X = q\X + q\Xm + p\X - p\Xm.$
Equations (\ref{eqn:pqeps1-hs}) and (\ref{eqn:pqeps2-hs}) show
that for $x\in S,$ one has
$u\X \geq (2-\epsilon) \{ q\X + q\Xm \}$
and 
$|u\Xm| \leq \epsilon \{ q\X + q\Xm \},$
whence 
\begin{align}
u&\X \geq \PARENS{\frac 2\epsilon -1} |u\Xm|, &&x\in S.
\label{eqn:u-unbalanced-hs}
\intertext{We also note that}
u&\X >0, &&x\in S.
\label{eqn:u-positive-hs}
\end{align}
Since $u$ has degree at most $d,$ we have by (\ref{eqn:r-orthog}) that
\begin{align*}
  \sum_{x\in S} \{ \psi\X u\X  + \psi\Xm u\Xm\}
= \sum_{S\cup -S} \psi\X u\X  
= 0,
\end{align*}
whence
\begin{align*}
\psi\X u\X  \leq |\psi\Xm u\Xm|
\end{align*}
for some $x\in S.$ At the same time, it follows from
(\ref{eqn:r-unbalanced-hs}),
(\ref{eqn:u-unbalanced-hs}), and (\ref{eqn:u-positive-hs}) that
\begin{align*}
 \psi\X u\X  > \delta\PARENS{\frac 2\epsilon -1}
   \lvert \psi\Xm u\Xm\rvert, && x\in S.
\end{align*}
We immediately obtain
$\delta(\{2/\epsilon\} -1)< 1, $
as was to be shown.
\end{proof}

\begin{remark}
The method of Theorem~\ref{thm:rational-degree-criterion-hs} amounts
to reformulating (\ref{eqn:u-unbalanced-hs}) and (\ref{eqn:u-positive-hs})
as a linear program and exhibiting a solution to its dual.  The
presentation above does not explicitly use the language of linear
programs or appeal to duality, however, because our goal is solely
to prove the correctness of our method and not its completeness.
\end{remark}

Using the criterion of Theorem~\ref{thm:rational-degree-criterion-hs} and our
preparatory work in Section~\ref{sec:preparatory}, we now establish a
key lower bound for the rational approximation of halfspaces within constant
error. 

\begin{theorem}
\label{thm:rtl-approx-of-halfspace}
Let $f\colon \{0,\pm1,\dots,\pm(3n+1)\}^{n+1}\to\moo$ be given by 
\begin{align*}
f(x) = \sign\PARENS{1 + \sum_{i=1}^{n+1}2^ix_i}.
\end{align*}
Then
\begin{align*}
R^+(f,n)=\Omega(1).
\end{align*}
\end{theorem}

\begin{proof}
Let $A_0,A_1,\dots,A_n$ be as defined in Theorem~\ref{thm:aux-thr-deg}. Put
$A=\bigcup A_i$ and define $\g\colon A\cup -A\to\moo$ by
\begin{align*}
\g(x) = 
\ccases{
(-1)^i,&x\in A_i,\\
(-1)^{i+1},&x\in -A_i.
}
\end{align*}
Then $\degthr(f)>2n$ by Theorem~\ref{thm:aux-thr-deg}. As a result,
Theorem~\ref{thm:gordan} guarantees the existence of a function
$\phi\colon A\cup -A\to\Re,$ not identically zero, such that 
\begin{align}
\phi(x)\g(x)\geq 0, \qquad x\in A\cup -A, \label{eqn:signrep}
\end{align}
and
\begin{align}
\sum_{A\cup -A} \phi(x)u(x) = 0 \label{eqn:phiorthog}
\end{align}
for every polynomial $u$ of degree at most $2n.$ Put
\begin{align*}
p(x) = \prod_{j=0}^{n-1}
\PARENS{ - 2^j\sqrt 2 + \sum_{i=1}^{n+1} 2^{i-1}x_i}
\end{align*}
and
\begin{align*}
\psi(x) = (-1)^n\{\phi(x)-\phi(-x)\}p(x).
\end{align*}
Define $S=A\setminus\psi^{-1}(0).$ Then $S\ne\varnothing$ 
by (\ref{eqn:signrep}) and the fact that 
$\phi$ is not identically zero on $A\cup -A.$
For $x\in S,$ we have $\psi(-x)\ne 0$ and
\begin{align*}
\frac{\lvert\psi(x)\rvert}{\lvert\psi(-x)\rvert}
 = \frac{\lvert p(x)\rvert}{\lvert p(-x)\rvert}
 > \PARENS{\prod_{i=1}^\infty \frac{2^{i/2}-1}{2^{i/2}+1}}^2 
 > \exp(-9\sqrt 2), 
\end{align*}
where the final step uses the bound $(a-1)/(a+1)>\exp(-2.5/a),$
valid for $a\geq\sqrt 2.$
It follows from (\ref{eqn:signrep}) and the definition of $p$ 
that $\psi$ is positive on $S.$ Hence,
\begin{align}
\psi(x) > \exp(-9\sqrt 2)\; \lvert\psi(-x)\rvert, \qquad x\in S.
\label{eqn:psibalance}
\end{align}
For any polynomial $u$ of degree no greater than $n,$ we infer from
(\ref{eqn:phiorthog}) that
\begin{align}
\sum_{S\cup -S} \psi(x)u(x) = 
  (-1)^n\sum_{A\cup -A} 
  \{\phi(x)-\phi(-x)\} u(x) p(x) = 0.
\label{eqn:psiorthog}
\end{align}
Since $f$ is positive on $S$ and negative on $-S,$
the proof is now complete in view of (\ref{eqn:psibalance}),
(\ref{eqn:psiorthog}), and Theorem~\ref{thm:rational-degree-criterion-hs}.
\end{proof}

We have reached the main result of this section, which extends
Theorem~\ref{thm:rtl-approx-of-halfspace} to any subconstant
approximation error and to halfspaces on the hypercube.

\begin{theorem}
\label{thm:main-halfspace}
Let $F\colon\moo^{m^2}\to\moo$ be given by
\begin{align*}
F(x) = \sign\PARENS{1 + \sum_{i=1}^m\sum_{j=1}^m 2^i x_{ij}}.
\end{align*}
Then for $d<m/14,$
\begin{align}
R(F,d) \geq 1 - 2^{-\Theta(m/d)}.
\label{eqn:approx-halfspace-lower}
\end{align}
\end{theorem}

Observe that Theorem~\ref{thm:main-halfspace} settles
the lower bounds in Theorem~\ref{thm:main-approx-hs} from the
Introduction.

\begin{proof}[Proof of Theorem~\textup{\ref{thm:main-halfspace}}.]
%
We may assume that $m\geq14,$ the claim being trivial otherwise.
Consider the function $G\colon
\moo^{(n+1)(6n+2)}\to\moo$
given by
\begin{align*}
G(x) = \sign\PARENS{1 + \sum_{i=1}^{n+1}
\; \sum_{j=1}^{6n+2}2^ix_{ij}},
\end{align*}
where $n=\lfloor (m-2)/6 \rfloor.$ For every $\epsilon>R^+(G,n),$
Proposition~\ref{prop:symm-rational} provides a rational function
$A$ on $\Re^{n+1}$ of degree at most $n$ such that, on the domain
of $G,$
\begin{align*}
\left\lvert G(x)-A\PARENS{\dots,\sum_{j=1}^{6n+2}
x_{ij},\dots}\right\rvert < \epsilon
\end{align*}
and the denominator of $A$ is positive.
Letting $f$ be the function in
Theorem~\ref{thm:rtl-approx-of-halfspace}, it follows that 
$| f(x_1,\dots,x_{n+1}) - A(2x_1,\dots,2x_{n+1}) |<\epsilon$
on the domain of $f,$ whence
\begin{align}
R^+(G,n) = \Omega(1).
\label{eqn:RnG}
\end{align}

We now claim that either $G(x)$ or $-G(-x)$ is a subfunction of
$F.$ For example, consider the following substitution for the variables
$x_{ij}$ for which $i>n+1$ or $j>6n+2$:
\begin{align*}
 &x_{mj}\gets (-1)^j,         &&(1\leq j\leq m),\\
 &x_{ij}\gets (-1)^{j+1},     &&(n+1<i<m,\quad 1\leq j\leq m),\\
 &x_{ij}\gets (-1)^{j+1},     &&(1\leq i\leq n+1, \quad \, j>6n+2).
\end{align*}
After this substitution, $F$ is a function of the remaining variables
$x_{ij}$ and is equivalent to $G(x)$ if $m$ is even, and to $-G(-x)$
if $m$ is odd. In either case, (\ref{eqn:RnG}) implies that
\begin{align}
R^+(F,n) = \Omega(1).
\label{eqn:R+nF}
\end{align}
Theorem~\ref{thm:error-boosting} shows that 
\begin{align*}
R(F,n/2)\leq 1 - \PARENS{\frac{1 - R(F,d)}{2}}^{1/\lfloor n/(2d)\rfloor}
\end{align*}
for $d=1,2,\dots,\lfloor n/2\rfloor,$ which yields 
(\ref{eqn:approx-halfspace-lower}) in light of
(\ref{eqn:rational-positive-denominator}) and (\ref{eqn:R+nF}).
\end{proof}

\section{Rational approximation of the majority function}
\label{sec:rational-approx-maj}

The goal of this section is to determine $R^+(\MAJ_n,d)$ for each
integer $d,$ i.e., to determine the least error to which a degree-$d$
multivariate rational function can approximate the majority function.
As is frequently the case with symmetric Boolean functions such as
majority, the multivariate problem of analyzing
$R^+(\MAJ_n,d)$ is equivalent to a univariate question. Specifically, given an
integer $d$ and a finite set $S\subset\Re,$ we define
\[ R^+(d,S) \,= \,\inf_{p,q}\, \max_{t\in S} 
   \left\lvert \sign t - \frac{p(t)}{q(t)} \right\rvert,\]
where the infimum ranges over $p,q\in P_d$ such that $q$ is positive on
$S.$  In other words, we study how well a rational
function of a given degree can approximate
the sign function over a finite support. We give a
detailed answer to this question in the following theorem:

\begin{theorem}[Rational approximation of \textsc{majority}]
\label{thm:R+}
Let $n,d$ be positive integers. Abbreviate 
$R=R^+(d,\{\pm1,\pm2,\dots,\pm n\}).$
For $1\leq d\leq\log n,$
\[\exp\BRACES{-\Theta\PARENS{\frac1{n^{1/(2d)}}}}
\leq R < \exp\BRACES{-\frac1{n^{1/d}}}. \]
For $\log n < d <n,$
\[ R=\exp\BRACES{-\Theta\PARENS{\frac{d}{\log (2n/d)}}}. \]
For $d\geq n,$ \[ R=0.\]
Moreover, the rational approximant is constructed explicitly in
each case.
\end{theorem}

Theorem~\ref{thm:R+} is the main result of this section. We establish
it in the next two subsections, giving separate treatment to 
the cases $d\leq\log n$
and $d>\log n$ (see Theorems~\ref{thm:rational-small-degree}
and~\ref{thm:rational-high-degree}, respectively). In the concluding
subsection, we give the promised proof that 
$R^+(d,\{\pm1,\dots,\pm n\})$ 
and $R^+(\MAJ_n,d)$ are essentially equivalent.


\subsection{Low-degree approximation}

We start by specializing the criterion of
Theorem~\ref{thm:rational-degree-criterion-hs} to the problem of
approximating the sign function on the set $\{\pm1,\pm2,\dots,\pm
n\}.$

\begin{theorem}
\label{thm:rational-degree-criterion-maj}
Let $d$ be an integer, $0\leq d\leq 2n-1.$ 
Fix a nonempty subset $S\subseteq\{1,2,\dots,n\}.$
Suppose that there exists a real $\delta\in(0,1)$
and a polynomial $r\in P_{\,2n-d-1}$ that vanishes on 
$\{-n,\dots,n\}\setminus (S\cup -S)$ and obeys
\begin{equation}
(-1)^t r(t) > \delta \lvert r(-t)\rvert,\qquad t\in S.
\label{eqn:r-unbalanced-maj}
\end{equation}
Then
\begin{align} 
R^+(d,S\cup -S) \geq \frac{2\delta}{1+\delta}.
\label{eqn:maj-criterion}
\end{align}
\end{theorem}

\begin{proof}
Define $f\colon S\cup-S\to\moo$ by $f(t)=\sign t.$ Define $\psi\colon S\cup
-S\to\Re$ by
$\psi(t) = (-1)^t {2n\choose n+t} r(t).$
Then (\ref{eqn:r-unbalanced-maj}) takes on the form
\begin{align}
\psi(t) > \delta \lvert \psi(-t)\rvert, \qquad t\in S.
\label{eqn:maj-psi-unbalanced}
\end{align}
For every polynomial $u$ of degree at most $d,$ we have
\begin{align}
\sum_{S\cup -S} \psi(t)u(t) = 
\sum_{t=-n}^n (-1)^t {2n\choose n+t} r(t)u(t)=0
\label{eqn:maj-psi-orthogonal}
\end{align}
by Fact~\ref{fact:comb}. Now (\ref{eqn:maj-criterion}) is immediate
from (\ref{eqn:maj-psi-unbalanced}), (\ref{eqn:maj-psi-orthogonal}),
and Theorem~\ref{thm:rational-degree-criterion-hs}.
\end{proof}

Using Theorem~\ref{thm:rational-degree-criterion-maj}, we will
now determine the optimal error in the approximation of the majority
function by rational functions of degree up to $\log n.$ The case
of higher degrees will be settled in the next subsection.

\begin{theorem}[Low-degree rational approximation of \textsc{majority}]
\label{thm:rational-small-degree}
Let $d$ be an integer, $1\leq d\leq \log n.$ Then
\[\exp\BRACES{-\Theta\PARENS{\frac1{n^{1/(2d)}}}}
\leq R^+(d,\{\pm1,\pm2,\dots,\pm n\}) 
< \exp\BRACES{-\frac1{n^{1/d}}}. \]
\end{theorem}

\begin{proof}
The upper bound is immediate from Newman's Theorem~\ref{thm:newman-approx}.
For the lower bound, put $\Delta = \lfloor n^{1/d}\rfloor\geq2$ and
$S = \{1,\Delta,\Delta^2,\dots,\Delta^d\}.$
Define $r\in P_{\,2n-d-1}$ by 
\[ r(t) = (-1)^n\prod_{i=0}^{d-1} (t-\Delta^i\sqrt\Delta) 
\prod_{i\in \{-n,\dots,n\}\setminus (S\cup-S)}  (t-i). \]
For $j=0,1,2,\dots,d,$ 
\begin{align*}
\frac{|r(\Delta^j)|}{|r(-\Delta^j)|} &= 
   \prod_{i=0}^{j-1} 
    \frac{\Delta^j - \Delta^i\sqrt\Delta}
         {\Delta^j + \Delta^i\sqrt\Delta}
	\;	 
   \prod_{i=j}^{d-1} 
    \frac{\Delta^i\sqrt\Delta - \Delta^j}
         {\Delta^i\sqrt\Delta + \Delta^j} 
  > \PARENS{\prod_{i=1}^\infty
   \frac{\Delta^{i/2}- 1}{\Delta^{i/2}+ 1}}^2\\
  &> \exp\BRACES{-5\sum_{i=1}^\infty
   \frac1{\Delta^{i/2}}} 
  >\exp\BRACES{- \frac{18}{\sqrt\Delta}},
\end{align*}
where we used the bound $(a-1)/(a+1)>\exp(-2.5/a),$ 
valid for $a\geq\sqrt2.$
Since $\sign r(t)=(-1)^t$ for $t\in S,$ we conclude that 
\[ (-1)^t r(t) >
\exp\BRACES{-\frac{18}{\sqrt\Delta}}
       \lvert r(-t)\rvert, \qquad
		 t\in S.\]
Since in addition $r$ vanishes on $\{-n,\dots,n\}\setminus (S\cup-S),$ 
we infer from Theorem~\ref{thm:rational-degree-criterion-maj} that
$R^+(d,S\cup -S)\geq\exp\{-18/\sqrt \Delta\}.$
\end{proof}

\subsection{High-degree approximation}

In the previous subsection, we determined the least error in
approximating the majority function by rational functions of degree
up to $\log n.$ Our goal here is to solve the case of higher degrees.

We start with some preparatory work. First, we need to
accurately estimate products of the form $\prod_i
(\Delta^i+1)/(\Delta^i-1)$ for all $\Delta>1.$ 
A suitable \emph{lower} bound was already given by 
Newman~\cite[Lem.~1]{newman64rational}:

\begin{lemma}[Newman] 
\label{lem:newman-exp-product}
For all $\Delta>1,$
\[ \prod_{i=1}^n \frac{\Delta^i+1}{\Delta^i-1} >
\exp\BRACES{\frac{2(\Delta^n-1)}{\Delta^n(\Delta-1)}}. \]
\end{lemma}

\begin{proof}
Immediate from the bound $(a+1)/(a-1)>\exp(2/a),$ which is 
valid for $a>1.$ 
\end{proof}

We will need a corresponding upper bound:

\begin{lemma}
\label{lem:infinite-product}
For all $\Delta>1,$
\[ \prod_{i=1}^\infty \frac{\Delta^i + 1}{\Delta^i - 1} < 
\exp\BRACES{\frac 4{\Delta-1}}.
\]
\end{lemma}

\begin{proof}
Let $k\geq0$ be an integer.
By the binomial theorem, $\Delta^i\geq(\Delta-1)i+1$ for integers $i\geq0.$ 
As a result, 
\begin{align*}
\prod_{i=1}^k \frac{\Delta^i + 1}{\Delta^i - 1} 
\leq \prod_{i=1}^k \frac{1}{i}\PARENS{i+ \frac{2}{\Delta-1}} \leq 
{k+\left\lceil\frac{2}{\Delta-1}\right\rceil \choose k}.
\end{align*}
Also,
\begin{align*}
\prod_{i=k+1}^\infty \frac{\Delta^i + 1}{\Delta^i - 1} 
<\prod_{i=0}^\infty \PARENS{1 + \frac{2}{(\Delta^{k+1}-1)\Delta^i}}
<\exp\BRACES{\frac{2\Delta}{(\Delta^{k+1}-1)(\Delta-1)}}.
\end{align*}
Setting $k=k(\Delta)=\left\lfloor \frac2{\Delta-1}\right\rfloor,$ 
we conclude that
\[ \prod_{i=1}^\infty \frac{\Delta^i + 1}{\Delta^i - 1}
< \exp\BRACES{\frac C{\Delta-1} }, \]
where 
\begin{align*}
C=\sup_{\Delta>1}\BRACES{ (\Delta-1)\ln
{k(\Delta)+\left\lceil\frac{2}{\Delta-1}\right\rceil \choose k(\Delta)}
 + \frac{2\Delta}{\Delta^{k(\Delta)+1}-1}}<4. 
 \tag*{\qedhere}
\end{align*}
\end{proof}

We will also need the following binomial estimate.

\begin{lemma}
\label{lem:binomial-ratio}
Put $p(t) = \prod_{i=1}^n \left(t-i-\frac12\right).$ Then
\[ \max_{t=1,2,\dots,n+1} \left\lvert \frac{p(-t)}{p(t)}\right\rvert \leq
\Theta(16^n). \]
\end{lemma}

\begin{proof}
For $t=1,2,\dots,n+1,$ we have
\begin{align*}
|p(t)| = \frac{(2t-2)! (2n-2t+2)!}{4^n (t-1)!(n-t+1)!},\quad
|p(-t)| = \frac{t!(2n+2t+1)!}{4^n (2t+1)!(n+t)!}.
\end{align*}
As a result,
\begin{align*}
\left\lvert\frac{p(-t)}{p(t)}\right\rvert 
 &= \frac t{2t+1}  \cdot
 \frac{ \displaystyle {2n+2t+1\choose 2t}{2n+1\choose n+t}}
      { \displaystyle {2t -2\choose t-1}{2n-2t+2 \choose n-t+1}}
 \leq \frac{\displaystyle \Theta\PARENS{\frac{2^{4n}}{\sqrt n}} 
             \Theta\PARENS{\frac{2^{2n}}{\sqrt n}}}
           {\displaystyle \Theta\PARENS{\frac{2^{2n}}n}},
\end{align*}
which gives the sought bound.
\end{proof}

Our construction requires one additional ingredient.

\begin{lemma}
\label{lemma:floors}
Let $n,d$ be integers, $1\leq d\leq n/55.$ Consider the
polynomial
$p(t) = \prod_{i=1}^{d-1} (t-d \Delta^i\sqrt \Delta),$
where $\Delta=(n/d)^{1/d}.$ Then
\[ \min_{j=1,\dots,d} \left| 
   \frac{p(\lfloor d\Delta^j\rfloor)}{p(-\lfloor d\Delta^j\rfloor)} \right|
	>\exp\BRACES{-\frac{4\ln3d}{\ln(n/d)} -\frac{8}{\sqrt \Delta-1}}.
\]
\end{lemma}

\begin{proof}
Fix $j=1,2,\dots,d.$ Then for each $i=1,2,\dots,j-1,$
\[ d\Delta^j - d\Delta^i\sqrt \Delta \geq
d\left(\Delta^{j-i-\frac12}-1\right)
                      \geq \frac12\, (j-i) \ln\frac nd,  \]
and thus
\begin{align}
\prod_{i=1}^{j-1} \PARENS{1 - \frac1{d\Delta^j - d\Delta^i\sqrt \Delta}}
&\geq  \exp\BRACES{-\frac{4}{\ln(n/d)}\sum_{i=1}^{j-1}
                  \frac{1}{j-i}
          }\nonumber\\
&\geq  \exp\BRACES{-\frac{4\ln3d}{\ln(n/d)} }.
\label{eqn:linearization}
\end{align}
For brevity, let $\xi$ stand for the final expression in
(\ref{eqn:linearization}). 
Since $1\leq d\leq n/55,$ we have $\lfloor d\Delta^j\rfloor
-d\Delta^{j-1}\sqrt \Delta>1.$ 
As a result,
\begin{align*}
\left|\frac{p(\lfloor d\Delta^j\rfloor)}{p(-\lfloor d\Delta^j\rfloor)}\right|
	&\geq \prod_{i=1}^{j-1} 
	\frac{d\Delta^j-1-d\Delta^i\sqrt \Delta}{d\Delta^j+d\Delta^i\sqrt \Delta}
	\;\; \prod_{i=j}^{d-1} 
	\frac{d\Delta^i\sqrt \Delta-d\Delta^j}{d\Delta^i\sqrt \Delta+d\Delta^j}\\
	&\geq \xi
	\prod_{i=1}^{j-1} 
	\frac{d\Delta^j-d\Delta^i\sqrt \Delta}{d\Delta^j+d\Delta^i\sqrt \Delta}
	\;\;\prod_{i=j}^{d-1} 
	\frac{d\Delta^i\sqrt \Delta-d\Delta^j}{d\Delta^i\sqrt \Delta+d\Delta^j}
	&&\text{by (\ref{eqn:linearization})}\\
	&> \xi
	\PARENS{\prod_{i=1}^\infty \frac{\Delta^{i/2} - 1}{\Delta^{i/2} +1}}^2\\
	&\geq \xi \exp\BRACES{-\frac 8{\sqrt \Delta-1}},
\end{align*}
where the last inequality holds by Lemma~\ref{lem:infinite-product}.
\end{proof}

We have reached the main result of this subsection.

\begin{theorem}[High-degree rational approximation of \textsc{majority}]
\label{thm:rational-high-degree}
Let $d$ be an integer, $\log n<d\leq n-1.$ Then
\begin{align*}
R^+(d,\{\pm1,\pm2,\dots,\pm n\}) =
\exp\BRACES{-\Theta\PARENS{\frac{d}{\log (2n/d)}}}.
\end{align*}
Also,
\[ R^+(n,\{\pm1,\pm2,\dots,\pm n\}) =0. \]
\end{theorem}

\begin{proof}
The final statement in the theorem follows at once by considering
the rational function $\{p(t)-p(-t)\}/\{p(t)+p(-t)\},$ where $p(t)
= \prod_{i=1}^n (t+i).$ 

Now assume that $\log n < d < n/55.$ Let 
\[ k = \left\lceil \frac{d}{\log (n/d)} \right\rceil, \qquad
\Delta=\PARENS{\frac nd}^{1/d}. \]
Define sets 
\begin{align*}
S_1 &= \{1,2,\dots,k\}, \\
S_2 &= \rule{0mm}{4mm} \{\lfloor d \Delta^i\rfloor\,:\, i=1,2,\dots,d \},\\
S_{\phantom{1}} &= \rule{0mm}{4mm} S_1\cup S_2.
\end{align*}
Consider the polynomial
\[ r(t) = (-1)^n r_1(t) r_2(t)    
  \prod_{i\in\{-n,\dots,n\}\setminus (S\cup -S)} (t-i),
\]
where 
\[ r_1(t) = 
   \prod_{i=1}^{k} \PARENS{t-i-\frac12}, \qquad
   r_2(t) = \prod_{i=1}^{d-1} (t-d \Delta^i\sqrt \Delta).
\]
We have:
\begin{align*}
\min_{t\in S} \left| \frac{r(t)}{r(-t)} \right| 
   &\geq 
   \min_{i=1,\dots,k+1} \left| \frac{r_1(i)}{r_1(-i)} \right| 
   \cdot
   \min_{i=1,\dots,d} \left| 
   \frac{r_2(\lfloor d\Delta^i\rfloor)}{r_2(-\lfloor d\Delta^i\rfloor)} \right|\\
   &> \exp\BRACES{-\frac{Cd}{\log (n/d)}}
\end{align*}
by Lemmas~\ref{lem:binomial-ratio}  and~\ref{lemma:floors},
where $C>0$ is an absolute constant. Since $\sign p(t)=(-1)^t$ for
$t\in S,$ we can restate this result as follows:
\[ (-1)^t r(t) > 
\exp\BRACES{-\frac{Cd}{\log (n/d)}}
|r(-t)|, \qquad t\in S. \]
Since $r$ vanishes on $\{-n,\dots,n\}\setminus (S\cup -S)$
and has degree $\leq 2n-1-d,$ we infer from
Theorem~\ref{thm:rational-degree-criterion-maj} that
$  R^+(d,S\cup-S) \geq \exp\BRACES{-Cd/\log (n/d)}. $
This proves the lower bound for the case $\log n < d < n/55.$

To handle the case $n/55\leq d \leq n-1,$ a different argument is
needed.  Let 
\[ r(t)=(-1)^n\,t\,\prod_{i=1}^d \PARENS{t-i-\frac12} 
\prod_{i=d+2}^n (t^2-i^2).  \]
By Lemma~\ref{lem:binomial-ratio},
there is an absolute constant $C>1$ such that 
\[ \left\lvert\frac{r(t)}{r(-t)}\right\rvert
  > C^{-d}, \qquad t=1,2,\dots,d+1.
\]
Since $\sign r(t)=(-1)^t$ for $t=1,2,\dots,d+1,$ we conclude that
\[ (-1)^t r(t) > C^{-d} |r(-t)|, \qquad t=1,2,\dots,d+1. \]
Since the polynomial $r$ 
vanishes on $\{-n,\dots,n\}\setminus \{\pm1,\pm2,\dots,\pm (d+1)\}$ 
and has degree $2n-1-d,$
we infer from Theorem~\ref{thm:rational-degree-criterion-maj} that
\[ R^+(d,\{\pm1,\pm2,\dots,\pm (d+1)\}) \geq C^{-d}. \]
This settles the lower bound for the case $n/55\leq d\leq n-1.$

It remains to prove the upper bound for the case $\log n<d\leq n-1.$ Here
we always have $d\geq2.$ 
Letting $k=\lfloor d/2\rfloor$ and $\Delta=(n/k)^{1/k},$ 
define $p\in P_{\,2k}$ by
\[ p(t) = \prod_{i=1}^k (t+i) \prod_{i=1}^k (t+k\Delta^i). \]
Fix any point $t\in \{1,2,\dots,n\}$ with $p(-t)\ne0.$ Letting $i^*$ be the
integer with $k \Delta^{i^*}<t<k \Delta^{i^*+1},$ we have:
\begin{align*}
\frac{p(t)}{|p(-t)|}
&> \prod_{i=0}^{i^*} \frac{k\Delta^{i^*+1} + k\Delta^i}{k\Delta^{i^*+1}-k\Delta^i}
  \prod_{i=i^*+1}^{k} \frac{k\Delta^i + k\Delta^{i^*}}{k\Delta^i-k\Delta^{i^*}}
\geq  \prod_{i=1}^k \frac{\Delta^i + 1}{\Delta^i-1}\\
&> \exp\BRACES{\frac{2(\Delta^k-1)}{\Delta^k(\Delta-1)}},
\end{align*}
where the last inequality holds by Lemma~\ref{lem:newman-exp-product}.
Substituting $\Delta=(n/k)^{1/k}$ and recalling that
$k\geq\Theta(\log n),$ we obtain
$p(t) > A |p(-t)|$ for $t=1,2,\dots,n,$ where 
\[ A = \exp\BRACES{\Theta\PARENS{\frac{k}{\log (n/k)}}}. \]
As a result, $R^+(2k,\{\pm1,\pm2,\dots,\pm n\}) \leq2A/(A^2+1),$
the approximant in question being 
\begin{align*}
\frac{A^2-1}{A^2+1}\cdot \frac{p(t)-p(-t)}{p(t)+p(-t)}.
\tag*{\qedhere}
\end{align*}
\end{proof}

\subsection{Equivalence of the majority and sign functions}

It remains to prove the promised equivalence of the majority and sign
functions, from the standpoint of approximating them by rational functions
on the discrete domain. We have:

\begin{theorem}
\label{thm:maj-vs-sign}
For every integer $d,$
\begin{align}
 R^+(\MAJ_n,d) &\leq R^+(d-2,\{\pm1,\pm2,\dots,\pm \lceil n/2\rceil\}),
 \label{eqn:upper-bound-R} \\
 R^+(\MAJ_n,d) &\geq R^+(d,\{\pm1,\pm2,\dots,\pm \lfloor n/2\rfloor\}).
 \label{eqn:lower-bound-R}
\end{align}
\end{theorem}

\begin{proof}
We prove (\ref{eqn:upper-bound-R}) first. Fix a degree-$(d-2)$
approximant $p(t)/q(t)$ to $\sign t$ on $S=\{\pm1,\dots,\pm\lceil
n/2\rceil\},$ where $q$ is positive on $S.$ For small $\delta>0,$
define
\[ A_\delta(t) = \frac{t^2p(t)-\delta}{t^2q(t)+\delta}. \]
Then $A_\delta$ is a rational function of degree at most $d$ whose
denominator is positive on $S\cup \{0\}.$ Finally, we have
$A_\delta(0)=-1$ and
\[ \lim_{\delta\to0} \max_{t\in S} |\sign t - A_\delta(t)|
=  \max_{t\in S} \left|\sign t - \frac{p(t)}{q(t)}\right|.
\]
Then $A_\delta(\frac12\sum (x_i+1) - \lfloor n/2\rfloor)$
is the desired approximant for $\MAJ_n(x_1,\dots,x_n).$
 

We now turn to the lower bound, (\ref{eqn:lower-bound-R}).
For every $\epsilon>R^+(\MAJ_n,d),$
Proposition~\ref{prop:symm-rational}
gives a univariate rational function $p(t)/q(t)$ of degree at most $d$ 
such that for all $x\in\moon,$ one has
\begin{align*}
\left\lvert \MAJ_n(x) - \frac{p(\sum x_i)}{q(\sum x_i)}\right\rvert
<\epsilon
\end{align*}
and  $q(\sum x_i)>0.$
Then
\begin{align*}
 \max_{t=\pm1,\pm2,\dots,\pm\lfloor n/2\rfloor}
\left| \sign t - \frac{p(2t+n-2\lfloor n/2\rfloor)}{q(2t+n-2\lfloor n/2\rfloor)}
\right|< \epsilon, 
\end{align*}
completing the proof of (\ref{eqn:lower-bound-R}).
\end{proof}

Note that (\ref{eqn:rational-positive-denominator})
and Theorems~\ref{thm:rational-small-degree},
\ref{thm:rational-high-degree}, and~\ref{thm:maj-vs-sign}
immediately imply Theorem~\ref{thm:main-approx-maj} from the
Introduction.

\begin{remark}
The proof that we gave for the upper bound, (\ref{eqn:upper-bound-R}),
illustrates a useful property of univariate rational approximants
$A(t)=p(t)/q(t)$ on a finite set $S.$ Specifically, given such an
approximant and a point $t^*\notin S,$ there exists an
approximant $A'$ with $A'(t^*)=a$ for any prescribed value $a$ and
$A'\approx A$ everywhere on $S.$  One such construction is 
\[ A'(t)=\frac{(t-t^*)p(t)+a\delta}{(t-t^*)q(t)+\delta}\]
for an arbitrarily small constant $\delta>0.$ Note that $A'$ has
degree only $1$ higher than the degree of the original approximant,
$A.$ This phenomenon is in sharp contrast to approximation by
polynomials, which do not possess this corrective ability.
\end{remark}

\section{Intersections of halfspaces
\label{sec:maj-maj}}

In this section, we prove our main theorems on the sign-representation
of intersections of halfspaces and majority functions.  
In the two subsections that follow, we give results for 
the threshold degree as well as \emph{threshold density}, another
key complexity measure of a sign-representation.

\subsection{Lower bounds on the threshold degree}

We start by formalizing the elegant observation due to Beigel et
al.~\cite{beigel91rational}, already described briefly in the
Introduction.

\begin{theorem}[Beigel, Reingold, and Spielman]
\label{thm:rational-is-possible}
Let $f\colon X\to\moo$ and $\g\colon Y\to\moo$ be given functions, where
$X,Y\subset \Re^n$ are finite sets. Let $d$ be an integer with 
$R^+(f,d) + R^+(\g,d)<1.$
Then
\begin{align*}
\degthr(f\wedge \g) \leq 2d.
\end{align*}
\end{theorem}

\begin{proof}
Fix rational functions $p_1(x)/q_1(x)$ and $p_2(y)/q_2(y)$ of degree
at most $d$ such that $q_1$ and $q_2$ are positive on $X$ and $Y,$
respectively, and
\begin{align*}
   \max_{x\in X} \left| f(x) - \frac{p_1(x)}{q_1(x)}\right| + 
   \max_{y\in Y} \left| \g(y) - \frac{p_2(y)}{q_2(y)}\right| < 1. 
\end{align*}
Then
\begin{align*}
 f(x)\wedge \g(y) \equiv \sign\{1 + f(x)+\g(y)\}
 \equiv \sign\BRACES{ 1 + \frac{p_1(x)}{q_1(x)} +
\frac{p_2(y)}{q_2(y)} }.
\end{align*}
Multiplying the last expression by the positive quantity $q_1(x)q_2(y),$
we obtain $f(x)\wedge \g(y) \equiv \sign\{ q_1(x)q_2(y) + p_1(x)q_2(y)
+ p_2(y)q_1(x)\}.$
\end{proof}

Recall that Theorem~\ref{thm:main-finite} gives an essentially
exact converse to Theorem~\ref{thm:rational-is-possible}.
We are now in a position to prove our main results on the
threshold degree.

\begin{theorem}[restatement of Theorems~\ref{thm:main-sign-hs} 
and~\ref{thm:main-sign-maj}]
\label{thm:intersections}
Consider the function $f\colon\moo^{n^2}\to\moo$ given by
\begin{align*}
f(x) = \sign\PARENS{1 + \sum_{i=1}^n\sum_{j=1}^n 2^i x_{ij}}.
\end{align*}
Let $\g\colon\moon\to\moo$ be the majority function on $n$ bits.
Then
\begin{align}
\degthr(f\wedge f)             &= \Omega(n),        \label{eqn:hs-hs}   \\
\degthr(\,\g\, \wedge\, \g\,)  &= \Omega(\log n).   \label{eqn:maj-maj}   
\end{align}
\end{theorem}

\begin{proof}
By Theorem~\ref{thm:main-halfspace}, we have $R^+(f,\epsilon n)\geq1/2$ for
some constant $\epsilon>0,$ which settles 
(\ref{eqn:hs-hs}) in view of Theorem~\ref{thm:main-finite}.

Analogously, Theorems~\ref{thm:R+} and~\ref{thm:maj-vs-sign} show 
that $R^+(\g,\epsilon\log n)\geq 1/2$ for some constant $\epsilon>0,$
which settles (\ref{eqn:maj-maj}) in view of
Theorem~\ref{thm:main-finite}.
\end{proof}

\begin{remark}
The lower bounds (\ref{eqn:hs-hs}) and (\ref{eqn:maj-maj}) are tight
and match the constructions due to Beigel et 
al.~\cite{beigel91rational}. These
matching upper bounds can be seen as follows.  By
Theorem~\ref{thm:upper-fnk}, we have $R^+(f,Cn)<1/2$ for some
constant $C>0,$ which shows that \mbox{$\degthr(f\wedge f)=O(n)$}
in view of Theorem~\ref{thm:rational-is-possible}.  Analogously,
Theorems~\ref{thm:R+} and~\ref{thm:maj-vs-sign} imply that
$R^+(\g,C\log n)<1/2$ for some constant $C>0,$ which shows that
\mbox{$\degthr(\g\wedge\g) =O(\log n)$} in view of
Theorem~\ref{thm:rational-is-possible}.

Furthermore, Theorem~\ref{thm:rational-is-possible} generalizes
immediately to conjunctions of $k=3$ and more functions.
In particular, the lower bounds in (\ref{eqn:hs-hs}) and
(\ref{eqn:maj-maj}) remain tight for intersections \mbox{$f\wedge
f\wedge\cdots\wedge f$} and \mbox{$\g\wedge \g\wedge\cdots\wedge
\g$} featuring any constant number of functions.
\end{remark}

We give one additional result, featuring the intersection of the
canonical halfspace with a majority function.

\begin{restatetheorem}{thm:main-sign-mixed}{\textsc{restated}}
Let $f\colon\moo^{n^2}\to\moo$ be given by
\begin{align*}
f(x) = \sign\PARENS{1 + \sum_{i=1}^n\sum_{j=1}^n 2^i x_{ij}}.
\end{align*}
Let $\g\colon\moo^{\lceil\sqrt n\rceil}\to\moo$ 
be the majority function on $\lceil\sqrt n\rceil$ bits.
Then
\begin{align}
\degthr(f\wedge\, \g\,)        &= \Theta(\sqrt n).
\label{eqn:f-and-g}
\end{align}
\end{restatetheorem}

\begin{proof}
We prove the lower bound first.
Let $\epsilon>0$ be a suitably small constant.
By Theorem~\ref{thm:main-halfspace}, we have $R^+(f,\epsilon \sqrt n)
\geq 1-2^{-\sqrt n}.$
By Theorems~\ref{thm:R+} and~\ref{thm:maj-vs-sign}, we have
$R^+(\g,\epsilon\sqrt n)\geq 2^{-\sqrt n}.$ In view of
Theorem~\ref{thm:main-finite}, these two facts imply that
$\degthr(f\wedge \g\,)=\Omega(\sqrt n).$

We now turn to the upper bound. It is clear that $R^+(\g,\lceil\sqrt
n\rceil)=0$ and $R^+(f,1)<1.$ It follows by
Theorem~\ref{thm:rational-is-possible} that $\degthr(f\wedge
\g)=O(\sqrt n).$
\end{proof}

\subsection{Lower bounds on the threshold density}

In addition to threshold degree, several other complexity
measures are of interest when sign-representing Boolean functions
by real polynomials.  One such complexity measure is
\emph{density,} i.e., the number of distinct monomials in any
polynomial that sign-represents a given function.  Formally, for
a given function $f\colon\moon\to\moo,$ the \emph{threshold
density} $\dns(f)$ is the minimum $k$ such that
\begin{align*}
f(x) \equiv \sign\PARENS{\sum_{i=1}^k\lambda_i\prod_{j\in S_i}
x_j}
\end{align*}
for some sets $S_1,\dots,S_k\subseteq\oneton$ and some reals
$\lambda_1,\dots,\lambda_k.$ We will show that intersections of
two halfspaces not only have high threshold degree but also high
threshold density.

We start with the conjunction of two majority functions.
Constructions in~\cite{beigel91rational} show
that the function $f(x,y)=\MAJ_n(x)\wedge \MAJ_n(y)$ can be
sign-represented by a linear combination of $n^{O(\log n)}$ monomials,
namely, the monomials of degree up to $O(\log n).$
Klivans and Sherstov~\cite[Thm.~1.2]{mlj07sq} complement
this with a lower bound of $n^{\Omega(\log
n/\log \log n)}$ on the number of distinct monomials needed. Our
next result improves this lower bound to a tight $n^{\Theta(\log n)}.$

\begin{theorem}
\label{thm:dns-majmaj}
Let $f\colon\moon\times\moon\to\moo$ be given by $f(x,y)=\MAJ_n(x_1,\dots,x_n)
\wedge \MAJ_n(y_1,\dots,y_n).$ Then
\begin{align*}
\dns(f) = n^{\Omega(\log n)}.
\end{align*}
\end{theorem}

\begin{proof}
Identical to the proof of Klivans and Sherstov~\cite[\S3.3,
Thm.~1.2]{mlj07sq}, with the only difference that
Theorem~\ref{thm:main-sign-maj} should be invoked in place of O'Donnell
and Servedio's earlier result~\cite{odonnell03degree} that 
$\degthr(f)=\Omega(\log n/\log\log n).$
\end{proof}

We will now derive an exponential lower bound on the threshold
density of the intersection of two halfspaces. For this, we
recall an elegant procedure for converting Boolean functions with
high threshold degree into Boolean functions with high threshold
density, discovered by Krause and
Pudl{\'a}k~\cite{krause94depth2mod}.
Their construction maps a given function
$f\colon\moon\to\moo$ to the function
$f^\op\colon(\moon)^3\to\moo$ given by 
\begin{align*}
f^\op(x,y,z) =
f(\dots, (\overline{z_i}\wedge x_i)\vee(z_i\wedge y_i), \dots).
\end{align*}
We have:

\begin{theorem}[{Krause and
Pudl{\'a}k~\cite[Prop.~2.1]{krause94depth2mod}}]
\label{thm:degree-length}
For every function $f\colon\moon\to\moo,$
\[ \dns(f^\op) \geq 2^{\degthr(f)}. \]
\end{theorem}

Another ingredient in our analysis is the following observation.

\begin{lemma}[Klivans and Sherstov~\cite{mlj07sq}]
\label{lem:ptf-reduction}
Let $f\colon\moon\to\moo$ be a given function.  Consider any
function $F\colon\moo^m\to\moo$ given by
$F(x) = f(\chi_1(x),\dots,\chi_n(x)),$ where each
$\chi_i$ is a parity function $\moom\to\moo$ or the negation of a
parity function. Then 
\[ \dns(f)\geq \dns(F). \]
\end{lemma}

\begin{proof}[Proof \textup{(Klivans and
Sherstov~\cite{mlj07sq})}.]
Immediate from the definition of threshold density and the fact
that the product of parity functions is another parity function. 
\end{proof}

We are now in a position to prove the desired result for
halfspaces.

\begin{theorem}
\label{thm:dns-hshs}
Let $f_n\colon\moo^{n^2}\to\moo$ be given by
\begin{align*}
f_n(x) = \sign\PARENS{1 +
\sum_{i=1}^n\sum_{j=1}^{\phantom{A}n\phantom{A}} 2^i x_{ij}}.
\end{align*}
Then
\begin{align}
\dns(f_n\wedge f_n)     
    &= \exp\{\Omega(n)\},  \label{eqn:dns-hs-hs}   \\
\dns(f_n\wedge \MAJ_{\lceil \sqrt n\rceil})  
    &= \exp\{\Omega(\sqrt n)\}.
	\label{eqn:dns-hs-maj}
\end{align}
\end{theorem}

\begin{remark}
\label{rem:f-g-ff-gg}
In the proof below, it will be useful to keep in mind the
following straightforward observation. Fix functions
$f,g\colon\mook\to\moo$ and define functions
$f',g'\colon\mook\to\moo$ by $f'(x)=-f(-x)$ and $g'(y)=-g(-y).$
Then we have $f'(x)\wedge g'(y)\equiv -(f(-x)\wedge g(-y)) f(-x)
g(-y),$ whence $\dns(f'\wedge g') \leq \dns(f\wedge
g)\dns(f)\dns(g)$ and 
thus
\begin{align}
\dns(f\wedge g) \geq \frac{\dns(f'\wedge g')}{\dns(f)\dns(g)}.
\label{eqn:f-g-ff-gg}
\end{align}
Similarly, we have
$f(x)\wedge g'(y)\equiv (f(x)\wedge g(-y))f(x),$ whence
\begin{align}
\dns(f\wedge g) \geq \frac{\dns(f\wedge g')}{\dns(f)}.
\label{eqn:f-g-f-gg}
\end{align}
To summarize, (\ref{eqn:f-g-ff-gg}) and (\ref{eqn:f-g-f-gg})
allow one to analyze the threshold density of $f\wedge g$ by
analyzing the threshold density of $f'\wedge g'$ or $f'\wedge g$
instead.  Such a transition will be helpful in our case.
\end{remark}

\begin{proof}
[Proof of Theorem~\textup{\ref{thm:dns-hshs}}.]
Put $m=\lfloor n/4\rfloor.$
The function
${f_m}^\op\colon(\moo^{m^2})^3\to\moo$ has the
representation
\begin{align*}
{f_m}^\op(x,y,z) = \sign\PARENS{1 +
\sum_{i=1}^m\sum_{j=1}^{\phantom{A}m\phantom{A}} 2^i
(x_{ij}+y_{ij}+x_{ij}z_{ij}-y_{ij}z_{ij})}.
\end{align*}
As a result, 
\begin{align*}
\dns(f_{4m}\wedge f_{4m})
  &\geq\dns({f_m}^\op\wedge{f_m}^\op) 
      &&\text{by Lemma~\ref{lem:ptf-reduction}}\\
  &=\dns((f_m\wedge f_m)^\op) \\
  &\geq 2^{\degthr(f_m\wedge f_m)}  
      && \text{by Theorem~\ref{thm:degree-length}} \\
  &\geq\exp\{\Omega(m)\}
      && \text{by Theorem~\ref{thm:intersections}.}
\end{align*}
By the same argument as in Theorem~\ref{thm:main-halfspace}, the
function $f_{4m}$ is a subfunction of $f_{n}(x)$ or $-f_n(-x).$
In the former case, (\ref{eqn:dns-hs-hs}) is immediate from the
lower bound on $\dns(f_{4m}\wedge f_{4m}).$ In the
latter case, (\ref{eqn:dns-hs-hs}) follows from the 
lower bound on $\dns(f_{4m}\wedge f_{4m})$ and
Remark~\ref{rem:f-g-ff-gg}.

The proof of (\ref{eqn:dns-hs-maj}) is entirely analogous.
\end{proof}

Krause and Pudl{\'a}k's method in Theorem~\ref{thm:degree-length}
naturally generalizes to linear combinations of
conjunctions rather than parity functions. In
other words, if a function $f\colon\moon\to\moo$ has threshold
degree $d$ and $f^\op(x,y,z)\equiv \sign(\sum_{i=1}^N\lambda_i
T_i(x,y,z))$ for some conjunctions $T_1,\dots,T_N$ of
the literals $x_1,y_1,z_1,\dots,x_n,y_n,z_n, \neg x_1,\neg
y_1,\neg z_1,\dots,\neg x_n,\neg y_n,\neg z_n,$ then $N\geq
2^{\Omega(d)}.$ With this remark in mind,
Theorems~\ref{thm:dns-majmaj} and~\ref{thm:dns-hshs} and their
proofs adapt easily to this alternate definition of density.

\addcontentsline{toc}{section}{Acknowledgments}
\section*{Acknowledgments}

I would like to thank Dima Gavinsky, 
Adam Klivans, Ryan O'Donnell, Ronald de Wolf,
and the anonymous reviewers
for their very helpful comments on an earlier version of this manuscript.  
I am also thankful to Ronald for telling me about applications of 
rational approximation to quantum query complexity.  I gratefully
acknowledge Scott Aaronson's tutorial on the polynomial
method, which motivated me to work on direct product theorems
for real polynomials.  This research was supported by Adam Klivans' NSF
CAREER Award and NSF Grant CCF-0728536.

{\small
\addcontentsline{toc}{section}{References}
\bibliographystyle{abbrv}
\bibliography{%
/Users/sasha/bib/general,%
/Users/sasha/bib/fourier,%
/Users/sasha/bib/cc,%
/Users/sasha/bib/learn,%
/Users/sasha/bib/my}

\begin{thebibliography}{10}

\bibitem{aaronson05postselection}
S.~Aaronson.
\newblock Quantum computing, postselection, and probabilistic polynomial-time.
\newblock {\em Proceedings of the Royal Society A}, 461(2063):3473--3482, 2005.

\bibitem{aaronson08tutorial}
S.~Aaronson.
\newblock The polynomial method in quantum and classical computing.
\newblock In {\em Proc.~of the 49th Symposium on Foundations of Computer
  Science (FOCS)}, page~3, 2008.

\bibitem{ABFKP:04}
M.~Alekhnovich, M.~Braverman, V.~Feldman, A.~R. Klivans, and T.~Pitassi.
\newblock Learnability and automatizability.
\newblock In {\em Proc.~of the 45th Symposium on Foundations of Computer
  Science (FOCS)}, pages 621--630, 2004.

\bibitem{allender89ac0tc0}
E.~Allender.
\newblock A note on the power of threshold circuits.
\newblock In {\em Proc.~of the 30th Symposium on Foundations of Computer
  Science (FOCS)}, pages 580--584, 1989.

\bibitem{ambainis05collision}
A.~Ambainis.
\newblock Polynomial degree and lower bounds in quantum complexity: {C}ollision
  and element distinctness with small range.
\newblock {\em Theory of Computing}, 1(1):37--46, 2005.

\bibitem{aspnes91voting}
J.~Aspnes, R.~Beigel, M.~L. Furst, and S.~Rudich.
\newblock The expressive power of voting polynomials.
\newblock {\em Combinatorica}, 14(2):135--148, 1994.

\bibitem{beigel93polynomial-method}
R.~Beigel.
\newblock The polynomial method in circuit complexity.
\newblock In {\em Proc.~of the Eigth Annual Conference on Structure in
  Complexity Theory}, pages 82--95, 1993.

\bibitem{beigel94perceptrons}
R.~Beigel.
\newblock Perceptrons, {$\mathsf{PP}$}, and the polynomial hierarchy.
\newblock {\em Computational Complexity}, 4:339--349, 1994.

\bibitem{beigel91rational}
R.~Beigel, N.~Reingold, and D.~A. Spielman.
\newblock {$\PP$} is closed under intersection.
\newblock {\em J. Comput. Syst. Sci.}, 50(2):191--202, 1995.

\bibitem{blum92trainingNN}
A.~L. Blum and R.~L. Rivest.
\newblock Training a 3-node neural network is {NP}-complete.
\newblock {\em Neural Networks}, 5:117--127, 1992.

\bibitem{BNRW05robust}
H.~Buhrman, I.~Newman, H.~R{\"o}hrig, and R.~de~Wolf.
\newblock Robust polynomials and quantum algorithms.
\newblock {\em Theory Comput. Syst.}, 40(4):379--395, 2007.

\bibitem{buhrman07pp-upp}
H.~Buhrman, N.~K. Vereshchagin, and R.~de~Wolf.
\newblock On computation and communication with small bias.
\newblock In {\em Proc. of the 22nd Conf. on Computational Complexity (CCC)},
  pages 24--32, 2007.

\bibitem{buhrman-dewolf02DT-survey}
H.~Buhrman and R.~de~Wolf.
\newblock Complexity measures and decision tree complexity: {A} survey.
\newblock {\em Theor. Comput. Sci.}, 288(1):21--43, 2002.

\bibitem{eremenko07sign}
A.~Eremenko and P.~Yuditskii.
\newblock Uniform approximation of $\mathrm{sgn}(x)$ by polynomials and entire
  functions.
\newblock {\em J. d'Analyse Math{\'e}matique}, 101:313--324, 2007.

\bibitem{gordan1873lp}
P.~Gordan.
\newblock {\" U}ber die {A}ufl{\" o}sung linearer {G}leichungen mit reellen
  {C}oefficienten.
\newblock {\em Mathematische Annalen}, 6:23--28, 1873.

\bibitem{hoyer-mosca-dewolf03and-or-tree}
P.~H{\o}yer, M.~Mosca, and R.~de~Wolf.
\newblock Quantum search on bounded-error inputs.
\newblock In {\em Proc. of the 30th International Colloquium on Automata,
  Languages, and Programming (ICALP)}, pages 291--299, 2003.

\bibitem{ioffe-tikhomirov68duality}
A.~D. Ioffe and V.~M. Tikhomirov.
\newblock Duality of convex functions and extremum problems.
\newblock {\em Russ. Math. Surv.}, 23(6):53--124, 1968.

\bibitem{khot-saket08hs-and-hs}
S.~Khot and R.~Saket.
\newblock On hardness of learning intersection of two halfspaces.
\newblock In {\em Proc.~of the 40th Symposium on Theory of Computing (STOC)},
  pages 345--354, 2008.

\bibitem{klivans-thesis}
A.~R. Klivans.
\newblock {\em A Complexity-Theoretic Approach to Learning}.
\newblock PhD thesis, Massachusetts Institute of Technology, 2002.

\bibitem{KOS:02}
A.~R. Klivans, R.~O'Donnell, and R.~A. Servedio.
\newblock Learning intersections and thresholds of halfspaces.
\newblock {\em J. Comput. Syst. Sci.}, 68(4):808--840, 2004.

\bibitem{KS01dnf}
A.~R. Klivans and R.~A. Servedio.
\newblock Learning {DNF} in time {$2^{\tilde O(n^{1/3})}$}.
\newblock {\em J. Comput. Syst. Sci.}, 68(2):303--318, 2004.

\bibitem{klivans-servedio06decision-lists}
A.~R. Klivans and R.~A. Servedio.
\newblock Toward attribute efficient learning of decision lists and parities.
\newblock {\em J.~Machine Learning Research}, 7:587--602, 2006.

\bibitem{KlivansServedio:04coltmargin}
A.~R. Klivans and R.~A. Servedio.
\newblock Learning intersections of halfspaces with a margin.
\newblock {\em J. Comput. Syst. Sci.}, 74(1):35--48, 2008.

\bibitem{mlj07sq}
A.~R. Klivans and A.~A. Sherstov.
\newblock Unconditional lower bounds for learning intersections of halfspaces.
\newblock {\em Machine Learning}, 69(2--3):97--114, 2007.

\bibitem{focs06hardness}
A.~R. Klivans and A.~A. Sherstov.
\newblock Cryptographic hardness for learning intersections of halfspaces.
\newblock {\em J. Comput. Syst. Sci.}, 75(1):2--12, 2009.

\bibitem{krause94depth2mod}
M.~Krause and P.~Pudl{\'a}k.
\newblock On the computational power of depth-$2$ circuits with threshold and
  modulo gates.
\newblock {\em Theor. Comput. Sci.}, 174(1--2):137--156, 1997.

\bibitem{KP98threshold}
M.~Krause and P.~Pudl{\'a}k.
\newblock Computing {B}oolean functions by polynomials and threshold circuits.
\newblock {\em Comput. Complex.}, 7(4):346--370, 1998.

\bibitem{KwekPitt:98}
S.~Kwek and L.~Pitt.
\newblock {PAC} learning intersections of halfspaces with membership queries.
\newblock {\em Algorithmica}, 22(1/2):53--75, 1998.

\bibitem{lee09formulas}
T.~Lee.
\newblock A note on the sign degree of formulas, September 2009.
\newblock Manuscript at arXiv/cc.CS.

\bibitem{minsky88perceptrons}
M.~L. Minsky and S.~A. Papert.
\newblock {\em Perceptrons: {A}n Introduction to Computational Geometry}.
\newblock MIT Press, Cambridge, Mass., 1969.

\bibitem{newman64rational}
D.~J. Newman.
\newblock Rational approximation to $|x|$.
\newblock {\em Michigan Math. J.}, 11(1):11--14, 1964.

\bibitem{nisan-szegedy94degree}
N.~Nisan and M.~Szegedy.
\newblock On the degree of {B}oolean functions as real polynomials.
\newblock {\em Computational Complexity}, 4:301--313, 1994.

\bibitem{odonnell03degree}
R.~O'Donnell and R.~A. Servedio.
\newblock New degree bounds for polynomial threshold functions.
\newblock In {\em Proc.~of the 35th Symposium on Theory of Computing (STOC)},
  pages 325--334, 2003.

\bibitem{OS-extremal-ptf}
R.~O'Donnell and R.~A. Servedio.
\newblock Extremal properties of polynomial threshold functions.
\newblock {\em J. Comput. Syst. Sci.}, 74(3):298--312, 2008.

\bibitem{paturi-saks94rational}
R.~Paturi and M.~E. Saks.
\newblock Approximating threshold circuits by rational functions.
\newblock {\em Inf. Comput.}, 112(2):257--272, 1994.

\bibitem{podolskii07perceptrons}
V.~V. Podolskii.
\newblock Perceptrons of large weight.
\newblock In {\em Proc.~of the Second International Computer Science Symposium
  in Russia (CSR)}, pages 328--336, 2007.

\bibitem{podolskii08perceptrons}
V.~V. Podolskii.
\newblock A uniform lower bound on weights of perceptrons.
\newblock In {\em Proc.~of the Third International Computer Science Symposium
  in Russia (CSR)}, pages 261--272, 2008.

\bibitem{RS07dc-dnf}
A.~A. Razborov and A.~A. Sherstov.
\newblock The sign-rank of {$\AC^0$}.
\newblock In {\em Proc. of the 49th Symposium on Foundations of Computer
  Science (FOCS)}, pages 57--66, 2008.

\bibitem{rivlin-book}
T.~J. Rivlin.
\newblock {\em An Introduction to the Approximation of Functions}.
\newblock Dover Publications, New York, 1981.

\bibitem{saks93slicing}
M.~E. Saks.
\newblock Slicing the hypercube.
\newblock {\em Surveys in Combinatorics}, pages 211--255, 1993.

\bibitem{sherstov07ac-majmaj}
A.~A. Sherstov.
\newblock Separating {$\AC^0$} from depth-2 majority circuits.
\newblock {\em SIAM J. Comput.}, 38(6):2113--2129, \setbox0=\hbox{aaa}2009.
\newblock Preliminary version in 39th STOC, 2007.

\bibitem{sherstov07quantum}
A.~A. Sherstov.
\newblock The pattern matrix method for lower bounds on quantum communication.
\newblock In {\em Proc. of the 40th Symposium on Theory of Computing (STOC)},
  pages 85--94, \setbox0=\hbox{bbb}2008.

\bibitem{dual-survey}
A.~A. Sherstov.
\newblock Communication lower bounds using dual polynomials.
\newblock {\em Bulletin of the {EATCS}}, 95:59--93,
  \setbox0=\hbox{blahblah}2008.

\bibitem{sherstov07symm-sign-rank}
A.~A. Sherstov.
\newblock The unbounded-error communication complexity of symmetric functions.
\newblock In {\em Proc. of the 49th Symposium on Foundations of Computer
  Science (FOCS)}, pages 384--393, \setbox0=\hbox{dddddd}2008.

\bibitem{sherstov09opthshs}
A.~A. Sherstov.
\newblock Optimal bounds for sign-representing the intersection of two
  halfspaces by polynomials.
\newblock Manuscript at arxiv/cs.CC, October \setbox0=\hbox{eee}2009.

\bibitem{shi-linear}
Y.~Shi.
\newblock Approximating linear restrictions of {B}oolean functions.
\newblock Manuscript, 2002.

\bibitem{siu-roy-kailath94rational}
K.-Y. Siu, V.~P. Roychowdhury, and T.~Kailath.
\newblock Rational approximation techniques for analysis of neural networks.
\newblock {\em IEEE Transactions on Information Theory}, 40(2):455--466, 1994.

\bibitem{Vempala:97}
S.~Vempala.
\newblock A random sampling based algorithm for learning the intersection of
  halfspaces.
\newblock In {\em Proc. of the 38th Symposium on Foundations of Computer
  Science (FOCS)}, pages 508--513, 1997.

\bibitem{vereshchagin95weight}
N.~K. Vereshchagin.
\newblock Lower bounds for perceptrons solving some separation problems and
  oracle separation of {$\mathsf{AM}$} from {$\mathsf{PP}$}.
\newblock In {\em Proc.~of the Third Israel Symposium on Theory of Computing
  and Systems (ISTCS)}, pages 46--51, 1995.

\end{thebibliography}
}

\end{document}